\documentclass[conference, 10pt]{IEEEtran}
\usepackage{amsthm}
\newtheorem{theorem}{Theorem}
\newtheorem{lemma}{Lemma}
\newtheorem{proposition}{Proposition}

\newtheorem{definition}{Definition}
\newtheorem{assumption}{Assumption}

\newtheorem{Rem}{Remark}

\usepackage{hyperref}
\usepackage{enumitem}
\newcounter{problem}
\newcounter{subproblem}[problem]
\newenvironment{problem}{\refstepcounter{problem}{\bfseries Problem~\theproblem}}{}
\newenvironment{subproblem}{%
    \refstepcounter{subproblem}%
    \begin{enumerate}[label=\protect\hyperlink{newpage}{\bfseries \theproblem\alph{subproblem}},ref=\theproblem\alph{subproblem},leftmargin=*]
     \item}{\end{enumerate}}

\usepackage{mathtools}

\usepackage{color}
\usepackage{xcolor}
\usepackage{colortbl}
\definecolor{purple}{RGB}{139, 0, 139}
\usepackage{manfnt}
\usepackage{url}

\newif\iftodo   
\todotrue
\newif\iftodoshort  
\todoshortfalse
\ifCLASSINFOpdf
  \usepackage[pdftex]{graphicx}
   \DeclareGraphicsExtensions{.pdf,.jpeg,.png,.tiff}
\else
   \usepackage[dvips]{graphicx}
   \DeclareGraphicsExtensions{.eps,.ps}
\fi
\usepackage{amssymb,amsmath}
\usepackage[mathscr]{euscript}
\DeclareMathOperator*\argmax{arg \, max}		
\usepackage{algorithmic}
\usepackage[ruled,vlined,commentsnumbered]{algorithm2e}

\usepackage{cite}
\usepackage{subcaption}
\graphicspath{ {./fig/} }
\newcommand{\Rmnum}[1]{\uppercase\expandafter{\romannumeral #1}}
\newcommand{\rmnum}[1]{\lowercase\expandafter{\romannumeral #1}}

\newcommand{\diag}{\mathop{\mathrm{diag}}}

\newcommand{\field}[1]{\mathbb{#1}}
\newcommand{\menge}[1]{\mathrm{#1}}
\newcommand{\emenge}[1]{\mathscr{#1}}
\newcommand{\set}[1]{\mathscr{#1}}

\newcommand{\operator}[1]{\mathrm{#1}}

\newcommand{\R}{{\field{R}}}
 
\newcommand{\NN}{{\field{N}}} 

\newcommand{\RN}{{\field{R}}_{+}}
\newcommand{\RP}{{\field{R}}_{++}}

\newcommand{\Ns}{{\emenge{N}}}
\newcommand{\Ks}{{\emenge{K}}}
\newcommand{\Fs}{{\emenge{F}}}
\newcommand{\Is}{{\emenge{I}}}
\newcommand{\Js}{\overline{\emenge{I}}}

\newcommand{\snr}{\operator{SNR}}
\newcommand{\sinr}{\operator{SINR}}
\newcommand{\rsrp}{{\operator{RSRP}}}

\newcommand{\Rec}{\operator{R}}

\newcommand{\ma}[1]{\boldsymbol{\mathbf{#1}}}
\newcommand{\ve}[1]{\boldsymbol{\mathbf{#1}}}

\newcommand{\vx}{\ve{x}}
\newcommand{\vb}{\ve{b}}
\newcommand{\va}{\ve{a}}

\newcommand{\vc}{\ve{c}}
\newcommand{\vh}{\ve{h}}
\newcommand{\vw}{\ve{w}}
\newcommand{\vv}{\ve{v}}
\newcommand{\vp}{\ve{p}}
\newcommand{\vq}{\ve{q}}
\newcommand{\vqt}{\overline{\ve{p}}}
\newcommand{\pt}{\overline{p}}
\newcommand{\vfo}{\overline{\ve{f}}}
\newcommand{\sfo}{\overline{f}}

\newcommand{\vps}{\ve{\psi}}
\newcommand{\vy}{\ve{y}}

\newcommand{\vf}{\ve{f}}

\newcommand{\vd}{\ve{d}}

\newcommand{\mA}{\ma{A}}
\newcommand{\mB}{\ma{B}}
\newcommand{\mLa}{\ma{\Lambda}}

\newcommand{\mI}{\ma{I}}
\newcommand{\mV}{\ma{V}}
\newcommand{\mVt}{\tilde{\ma{V}}}

\newcommand{\mH}{\ma{H}}

\newcommand{\mD}{\ma{D}}

\usepackage[utf8]{inputenc}

\newcommand{\ul}{\text{UL}}
\newcommand{\dl}{\text{DL}}
\newcommand{\ext}{\text{ext}}

\newcommand{\fix}{\text{Fix}}

\newcommand{\ulul}{\text{UL}\leftarrow\text{UL}}
\newcommand{\dlul}{\text{UL}\leftarrow\text{DL}}
\newcommand{\uldl}{\text{DL}\leftarrow\text{UL}}
\newcommand{\dldl}{\text{DL}\leftarrow\text{DL}}
\newcommand{\cosl}[1]{}
\newcommand{\resl}[1]{}

\begin{document}
%
\title{Dynamic Joint Uplink and Downlink Optimization for Uplink and Downlink Decoupling-Enabled 5G Heterogeneous Networks}
\author{
\IEEEauthorblockN{Qi Liao\IEEEauthorrefmark{1}, Danish Aziz\IEEEauthorrefmark{1}, and S\l awomir Sta\'{n}czak\IEEEauthorrefmark{2}}\\
\IEEEauthorblockA{\IEEEauthorrefmark{1}Bell Labs, Nokia, Germany,  
\url{{qi.liao, danish.aziz}@nokia.com}}\\
\IEEEauthorblockA{\IEEEauthorrefmark{2}Technische Universit\"{a}t Berlin,  Germany,  
\url{slawomir.stanczak@tu-berlin.de}}
}

\maketitle
\begin{abstract}
The concept of user-centric and personalized service in the fifth generation (5G) mobile networks encourages technical solutions such as dynamic asymmetric uplink/downlink resource allocation and elastic association of cells to users with decoupled uplink and downlink (DeUD) access.
In this paper we develop a joint uplink and downlink optimization algorithm for DeUD-enabled wireless networks for adaptive joint uplink and downlink bandwidth allocation and power control, under different link association policies. 
%
%
Based on a general model of inter-cell interference, we propose a three-step optimization algorithm to jointly optimize the uplink and downlink bandwidth allocation and power control, using the fixed point approach for nonlinear operators with or without monotonicity, to maximize the minimum level of quality of service satisfaction per link, subjected to a general class of resource (power and bandwidth) constraints.
We present numerical results illustrating the theoretical findings for network simulator in a real-world setting, and show the advantage of our solution compared to the conventional proportional fairness resource allocation schemes in both the coupled uplink and downlink (CoUD) access and the novel link association schemes in DeUD. 
\end{abstract}
\begin{IEEEkeywords}
5G,  flexible duplex, decoupled uplink and downlink access, resource allocation 
\end{IEEEkeywords}

\section{Introduction}
The high rate of growth in global mobile data traffic drives the operators to set foot on the path of delivering the fifth generation (5G) of mobile networks, for user-centric and personalized service supporting diverse and often conflicting key performance indicators (KPIs), such as high-speed, low-latency, high reliability, high mobility, and low cost/energy consumption. 

In the 5G era, the evolution of heterogeneous networks (HetNets) results in cell densification with cells of different sizes. Due to the time- and spatial-dependent service requirements and traffic patterns, it is expected to have time-varying asymmetric traffic load in both uplink (UL) and downlink (DL) in different cells (as shown in Fig. \ref{fig:Cell_UL_DL}). 
Many optimization strategies are designed to provide seamless coverage and quality of service (QoS) in DL, while little interest has been shown in UL. However, the importance of UL grows along with the evolution of social networking and information/resource sharing system. 
Therefore, it is of great interest to develop a general framework for joint UL/DL optimization of resource allocation and power control, to adapt to the traffic asymmetry between UL and DL. 

Apart from dynamic UL/DL resource splitting, flexible UL/DL traffic distribution among the cells with different transmission ranges is also crucial for improvement of joint UL/DL performance. As proposed in \cite{andrews2013seven,boccardi2014five}, one way to enable the flexible UL/DL traffic distribution is to allow the user terminal to be associated to two different radio access nodes in UL and DL, respectively. Such a decoupled uplink and downlink (DeUD) access has the potential benefits including improvement of performance in UL (without  degradation of performance in DL), reduction of energy consumption in mobile terminal, and network load balancing.

The joint UL/DL optimization framework can benefit from the user-centric context-aware communication environment in 5G networks. More specifically, this includes dynamic splitting resources and distributing network traffic between UL and DL, based on the awareness of the heterogeneity of UL and DL channel conditions and traffic demands. 

The focus of this paper is to develop a general model of joint UL/DL interference, and to design a joint UL/DL optimization algorithm for
adaptive UL/DL bandwidth allocation and power control under different association policies for DeUD-enabled wireless networks. The objective is to optimize the minimum level of QoS satisfaction across all service links, using the fixed point approach for nonlinear operators with or without monotonicity.

\begin{figure}[t]
\centering
\includegraphics[width=1\columnwidth]{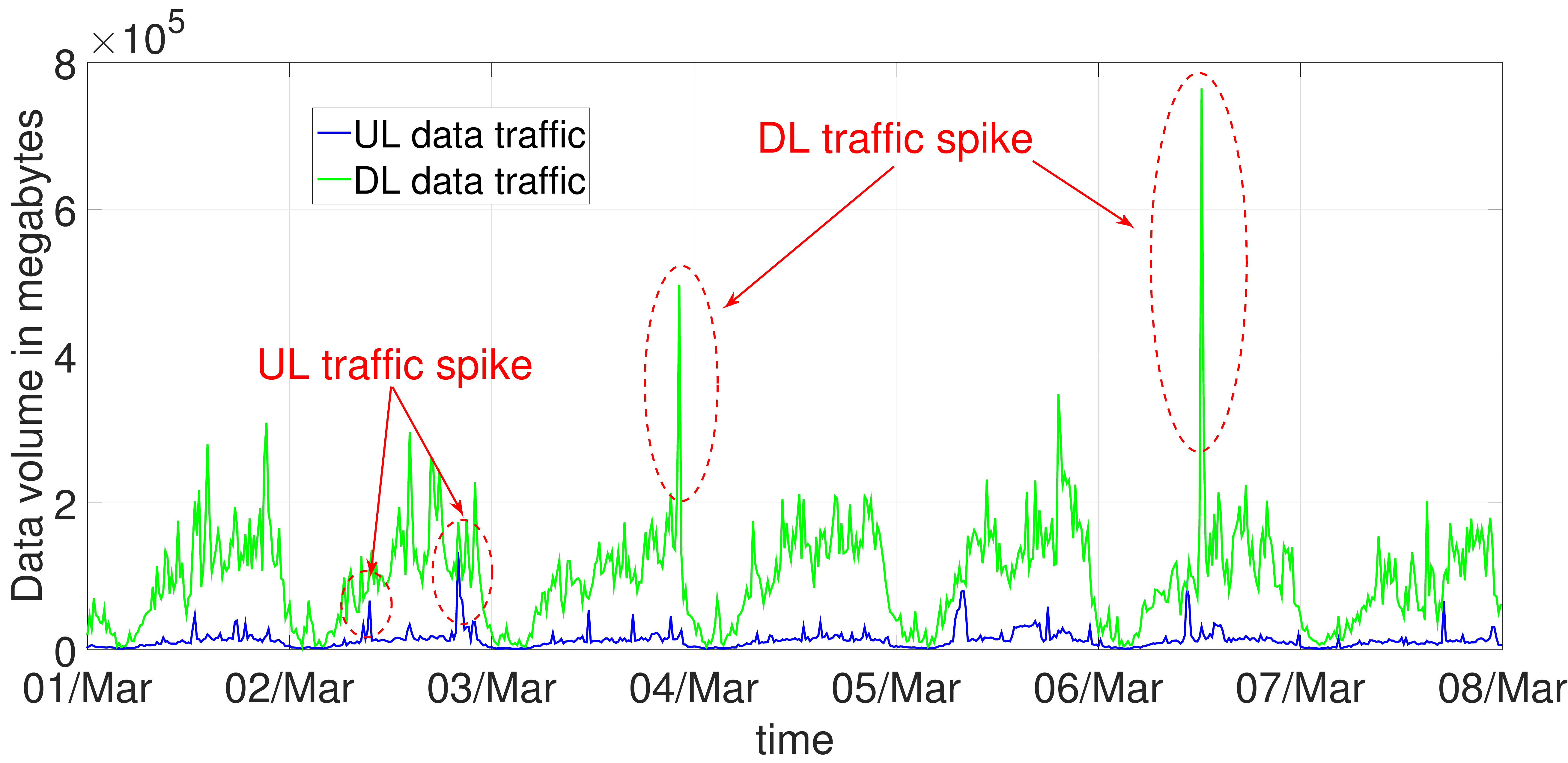}
\caption{Time-varying UL and DL data traffic volume (aggregated every 15 minutes) for a week from Mar. 01 to Mar. 08, 2015 in a spatial grid in Rome, Italy. Data source from Telecom Italia's Big Data Challenge \cite{TelecomItalia}.}
\label{fig:Cell_UL_DL}
\end{figure}
\subsection{Related Work}
%
\subsubsection{Joint Uplink and Downlink Optimization}\label{subsubsec:JointULDL}
Although much work has been done on the joint UL/DL resource allocation in conventional network with coupled uplink and downlink (CoUD) association \cite{su2007joint,schubert2005iterative,el2012stable,abdel2011optimization,chiang2009adaptive,kaur2010method}, to the best of the author's knowledge, none of the authors has worked on the problem for the next-generation networks with disruptive architectural design such as DeUD. For example, both of authors in \cite{chen2012joint} and \cite{liu2015backhaul} propose user association schemes in CoUD. The goal of the former is to jointly maximize the system
capacity in DL and to minimize transmitting power consumption in UL, while the aim of the latter is to minimize the sum of UL and DL average traffic delay and to reduce the overall UL and DL power consumption. 

Another restriction of the existing works is that they concern with the intra-cell communication either in the standard OFDMA-based networks or in the static or dynamic TDD-based networks. For example, the authors in \cite{el2012stable}  proposed a subcarrier allocation algorithm to maximize a utility function that captures the joint
UL/DL QoS requirements, by formulating the problem as a two-sided stable matching game. In \cite{kim2009joint}, a network utility maximization framework is proposed to solve the joint UL/DL resource allocation problem considering systems with FDD or static
TDD through the user-level satisfaction.

\subsubsection{Decoupled Uplink and Downlink Access}\label{subsubsec:DeUD}
The concept of downlink/uplink decoupling (DUDe\footnote{In this paper, we use a different term DeUD for \lq\lq decoupled uplink/downlink\rq\rq,  in consistency with the term CoUD for \lq\lq coupled uplink/downlink\rq\rq.}) is introduced in \cite{andrews2013seven,astely2013lte,boccardi2014five,boccardi2015decouple}. The recent contributions can be classified in three groups. 
  
The first group of articles focuses on the architectural design and realization. The pioneering contributions \cite{boccardi2014five,boccardi2015decouple} identify and explain some key arguments in favor of DUDe based on a blend of theoretical, experimental, and logical arguments.

The second group proposes varies link association policies and show the performance gain with simulations based on LTE field trial network. In \cite{elshaer2014downlink}, the notion of DUDe is studied, where the downlink cell association is based on the downlink received power while the uplink is based on the pathloss. The follow-up work \cite{elshaer2014load} considers the cell-load as well as the available backhaul capacity during the association process. One other idea for range extension of small cells in UL is to add a cell selection offset to the reference signals, to increase the priority of the small cells to be selected \cite{europe2008range}.

Last but not least, the third group of articles studies on the analytical evaluation of a predefined association policy. 
The work in \cite{smiljkovikj2014capacity,smiljkovikj2015analysis} focuses on the analytical characterization of the decoupled access by using the framework of stochastic geometry,  applying the same association criteria as in \cite{elshaer2014downlink}.
In \cite{DBLP:journals/corr/SinghZA14}, the authors propose  a  model to characterize the uplink SINR and rate distribution  as a function of the association rules (assuming weighted pathloss for both uplink and downlink association) and power control parameters (assuming fractional  pathloss-inversion  based  power control).

\subsubsection{Fixed-Point Based Framework for Max-Min Utility Maximization}
Yates \cite{Yates95a,Yates95b} proposed a framework of power control that is based on the notions of positivity, monotonicity, and scalability of standard interference functions (for details see Appendix \ref{subset:SIF}), to solve the SIR balancing problem. Since then, the framework of interference calculus is widely studied for the utility maximization involving only power and rate control. 
In \cite{UlY98,luo2003probability,luo2005standard}, the authors extend Yates' framework to stochastic power control algorithms.

The authors in \cite{chiang2004balancing,BoSchStWi05a,schubert05lST,boche2008structure,stanczak2009fundamentals} studied the max-min utility fairness problem with deterministic interference function involving power or rate control, and characterized the feasibility using the Perron-Frobenius theorem \cite{frobenius1912matrizen}.
Recent work \cite{zheng2014optimal,hong2014unified} leverages the nonlinear Perron-Frobenius theory \cite{lemmens2012nonlinear} and overcome the non-convexity or non-monotonicity in special cases of wireless utility maximization. In \cite{zheng2014optimal}, examples of SINR- or reliability-related non-convex utility optimization were introduced involving power control only. In \cite{hong2014unified}, the author proposes a general framework that enables rigorous treatment of nonlinear monotonic constraints in the utility fairness resource allocation problems.

In \cite{nuzman2007contraction}, the properties of standard interference function are re-examined from a contraction mapping point of view, where the convergence to a unique fixed point follows by a version of the Banach fixed point theorem \cite{smart1980fixed}. The theory provided in \cite{nuzman2007contraction} can be extended to certain non-monotonic functions.

\subsubsection{Interference Model Based on Power and Load Coupling}
The above-mentioned work typically addresses the inter-cell interference model with power coupling. In \cite{siomina2012load,Cavalcante14,ho2014data}, the authors consider the inter-cell interference characterized by the load coupling model, where cell load measures the average level of resource usage in the cell and implies the probability of generating interference from a transmitter to a receiver in OFDM sytsems. 
The interaction between power and load coupling are analyzed in \cite{cavalcante2014power,ho2014power}. The authors in \cite{cavalcante2014power} derive an interference mapping having as its fixed point the power allocation including a given load profile. The authors in \cite{ho2014power} address an energy minimization problem, and prove that operating at fill load is optimal in minimizing the sum energy.

\subsection{Contribution}
The main contributions of this paper are listed as follows.

We consider the next-generation wireless HetNets with disruptive architectural design with respect to dynamic splitting of UL/DL resource and link association. A common set of resource blocks are considered joint resource for both UL and DL services, and adaptive resource partitioning between UL and DL is enabled
to adapt to the link-specific traffic demand. The decoupled UL and DL access is further introduced to adapt to the link-specific channel condition (as shown in Fig. \ref{fig:ULDLSplit}).

We introduce a general model of inter-cell interference for joint UL/DL system. It includes the inter-link interference between UL and DL and is power and load coupling-aware. 
%
%
We then develop a framework involving a fixed-point class with nonlinear contraction operators (mainly motivated by the work in \cite{nuzman2007contraction}), and an optimizer for the utility of QoS satisfaction level, subjected to a general class of resource constraints. A three-step optimization algorithm is proposed, to find the local optimum of the joint variables bandwidth allocation and power spectral density on a per-link basis, corresponding to the different link association policies. 

To adapt the framework to the practical interest, we extend the work to cover the following aspects: 1) per-transmitter power control instead of per-link power control, and 2) energy efficient power control.

%
The rest of the paper is organized as follows. In Section \ref{sec:SysModel} we introduce some basic notations and system model. In Section \ref{ProbFormulation}, we present the utility fairness problem and its decomposition into two subproblems. The solution to the subproblem of adaptive joint UL/DL bandwidth allocation is provided in Section \ref{sec:ResourceAllo}, while of joint UL/DL power control (including the extension to the per-transmitter power control and energy efficient power control) in Section \ref{sec:PowerControl}. The joint algorithm to solve the main optimization problem is summarized in Section \ref{sec:algorithm}. 
The performance of the proposed algorithms are evaluated numerically in Section \ref{sec:simu}. We conclude the paper in Section \ref{sec:concl}.

\section{System Model}\label{sec:SysModel}
In this paper, we use the following standard definitions. The nonnegative and positive orthant in $k$ dimensions are denoted by $\RN^k$ and $\RP^k$, respectively. Let $\vx\leq\vy$ denote the component-wise inequality between two vectors $\vx$ and $\vy$. And let $\diag(\vx)$ denote a diagonal matrix with the elements of $\vx$ on the main diagonal. For a function $\vf:\R^k\to\R^k$, $\vf^n$ denotes the $n$-fold composition so that $\vf^n=\vf\circ\vf^{n-1}$. The $k\times k$ identity matrix is denoted by $\mI_k$ and the $n\times k$ zero matrix is denoted by $\ma{0}_{n\times k}$. The $k$-dimensional all-ones (all-zeros) vector is denoted by $\ve{1}_k$ ($\ve{0}_k$). The horizontal concatenation of two matrices $\mA \in \R^{n\times k}$, $\mB\in\R^{n\times l}$ is written as $[\mA \ | \ \mB]$, while the vertical concatenation of two matrices $\mA \in \R^{n\times k}$, $\mB\in\R^{m\times k}$ is written as $[\mA ; \mB]$. The cardinality of set $\set{A}$ is denoted by $|\set{A}|$.
The notation that will be used in this paper is summarized in Table \ref{tab:notation}.

We consider an OFDM-based wireless system consisting of a set of base stations (BSs) $\Ns$ with $| \set{N}| = N$ and a set of user equipments (UEs) $\Ks$ with $|\set{K}| = K$. We drop the usual assumption in wireless system design that UL and DL transmissions are associated with the same BS, and assume that they can be split. Let the UL(DL) cell-UE association matrix be denoted by $\mA^{\ul}\in\{0,1\}^{N\times K}$($\mA^{\dl}\in\{0,1\}^{N\times K}$).

We assume the reciprocal UL and DL channels. The set of all links (including ULs and DLs) is denoted by $\overline{\Ks} \coloneqq \Ks^{\ul}\cup \Ks^{\dl}$, where $\Ks^{\ul}$ and $\Ks^{\dl}$ are the sets of ULs and DLs, respectively. Because ULs and DLs have different transmitters and receivers, we have that $\Ks^{\ul}\cap\Ks^{\dl} = \emptyset$.  Without loss of generality, we assume that $|\Ks^{\ul}|=|\Ks^{\dl}| = K$ and $|\Ks| = 2K$.
We define the power spectral density (PSD) to be the transmit power assigned per resource block (RB),
and we use $\vp^{\ul}\in\RN^{K}$ and $\vp^{\dl}\in\RN^{K}$ to denote the vectors of uplink and downlink PSDs, respectively. Accordingly, $\vw^{\ul}\in[0,1]^{K}$ is used to denote fraction of the allocated RBs (normalized by dividing the number of allocated RBs by the total number of the available RBs), while $\vw^{\dl}\in[0,1]^{K}$ is the vector for such fractions in the downlink.
We collect $\vp^{\ul}$ and $\vp^{\dl}$ in one power vector  $\vp \coloneqq [\vp^{\ul}; \vp^{\dl}]\in\RN^{2K}$, and collect $\vw^{\ul}$ and $\vw^{\dl}$ in $\vw \coloneqq [\vw^{\ul};\vw^{\dl}]\in[0,1]^{2K}$. Let the total number of the RBs be denoted by $W_0$.

We consider the flexible duplex mode that allows UL and DL transmissions to share a joint set of RBs and to dynamically split between the RBs allocated to UL and DL. The split ratio is time-variant and cell-specific. 
Flexible duplex mode is proposed as the next step of FDD/TDD convergence in 5G networks \cite{alliance20155g,dahlman20145g}. The rapid evolution of subband-based splitting and filtering \cite{Zhang2015FilteredOFDM} and full duplex technology \cite{bharadia2014robust} makes dynamic splitting of spectrum allocated to UL and DL realizable in the near future.
%
%
%
The main drawback results from the need for coping with more intricate inter-cell interference structures: the interference is not only restricted to UL-to-UL and DL-to-DL interference, but also includes the inter-link interference between UL and DL, as shown in Fig. \ref{fig:DynamicFDDTDD_2}.

\begin{Rem}[Adaptation to Dynamic TDD]
Although in this paper the system model and optimization algorithm are developed based on forward-looking assumption of flexible duplex, they can be well adapted to more practical system with dynamic TDD configuration, by interpreting $\vw^{\ul}$ and $\vw^{\dl}$ as fraction of time frames dedicated to UL and DL, respectively. In this incident,  we can see the resource on the horizontal axis in Fig.\ref{fig:DynamicFDDTDD_2} as time frames instead of frequency subbands, and the inter-cell inter-link interference appears in the central frames that are used for UL transmission in BS $j$, while for DL transmission in another BS $i$.
\end{Rem}

%
%
\begin{figure}[t]
\centering
\includegraphics[width=.8\columnwidth]{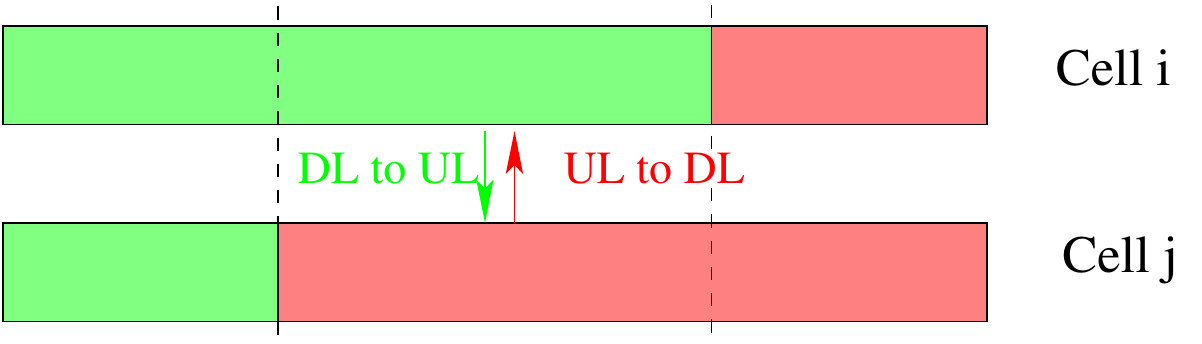}
\caption{Inter-cell inter-link interference between UL (red) and DL (green).  The guard band is not displayed.}
\label{fig:DynamicFDDTDD_2}
\end{figure}

\begin{table}[t]
\centering
\caption{NOTATION SUMMARY}
\begin{tabular}{|c|c|}
\hline
$\Ns$ & set of (macro and pico) BSs  \\
$\Ks$ & set of UEs \\
$\Ks^{\ul}$ ($\Ks^{\dl}$) & set of ULs (DLs)\\
$\overline{\Ks}$ & set of all service links\\
$\ma{A}^{\ul}$ ($\ma{A}^{\dl}$) & BS assignment matrix for ULs (DLs)\\
$\ma{A}$ & BS assignment matrix for all service links\\
$b_k^{\ul}$ ($b_k^{\dl}$) &  BS associated to the $k$th UL (DL)\\
$\vp^{\ul}$ ($\vp^{\dl}$) & PSD for ULs (DLs)\\
$\vp$ & PSD for all service links\\
$\vq^{\dl}$ & cell-specific PSD in DL\\
$\overline{\vp}$ & per-transmitter PSD\\
$\vw^{\ul}$ ($\vw^{\dl}$) &  fraction of allocated RBs for ULs (DLs)\\
$\vw$ & fraction of allocated RBs for all service links\\
$d_l$ & traffic demand (bit rate) of the $l$th link, $l\in\overline{\set{K}}$\\ 
$r_l$ & spectral efficiency of the $l$th link, $l\in\overline{\set{K}}$\\
$W_0$ & total number of RBs\\
$\ma{V}$ & link gain coupling matrix \\
$\tilde{\ma{V}}$ & link gain coupling matrix without intra-cell interference \\
$g_1(\vw)$ & constraint function implying the constraint on load\\
$g_2(\vw,\vp)$ & contraint function implying the contraint on transmit power\\
$\lambda$ & objective utility \\
$\Pi$ & set of link association policies\\
\hline
\end{tabular}
\label{tab:notation}
\end{table}

\subsection{Constrained Per-Cell Load and Per-Transmitter Power}
Since the UL and DL transmissions share a common set of resource blocks, we define the {\it cell load} to be the fraction of the total occupied frequency resource (in UL and DL) per cell. We collect the per-cell loads in a vector $\ve{\nu} \coloneqq \mA\vw\in[0,1]^N$, where $\mA \coloneqq \left[\mA^{\ul} \  | \ \mA^{\dl}\right]\in\{0,1\}^{N\times 2K}$ is the binary association matrix. Since the per-cell load is bounded above by $1$, we have
\begin{equation}
\RN^{2K}\to [0,1]: \ g_1(\vw) \coloneqq \|\mA\vw\|_{\infty}\leq 1.
\label{eqn:loadConstraint}
\end{equation}
This implies that for each cell, the sum of the fractions of allocated RBs for both UL and DL is constrained, i.e., $\forall n\in\Ns$ we have $\sum_ {k\in\Ks} \left(a^{\ul}_{n,k}w^{\ul}_k+a^{\dl}_{n,k}w^{\dl}_k\right)\leq 1$.

Let $\vp^{\ul}_{\max}\in\RP^{K}$ and $\vq_{\max}^{\dl}\in\RP^{N}$ denote the maximum UL transmit power per UE and the maximum DL transmit power per BS for the whole frequency band, respectively. 
Note that the maximum transmit power of a macro BS and a pico BS can vastly differ from each other in HetNets.
We define the extended maximum power vector by $\vp^{\ext}_{\max} \coloneqq [\vp_{\max}^{\ul};\vq_{\max}^{\dl}]\in\RP^{K+N}$ and the extended assignment matrix for transmitter-to-link association by $\mA^{\ext} \coloneqq [\mI_K \  | \ \ma{0}_{K\times K}; \ma{0}_{N\times K} \  | \ \mA^{\dl}]\in \{0,1\}^{(K+N) \times 2K}$. 
The per-transmitter (including both UEs and BSs) power constraints imply that 
\begin{align}
&\RN^{2K} \times \RN^{2K}\to\RN: \nonumber \\
&g_2(\vw, \vp): = W_0\|\diag(\vp^{\ext}_{\max})^{-1} \mA^{\ext}\diag(\vw)\vp\|_{\infty} \leq 1,
\label{eqn:powConstraint}
\end{align}
which is equivalent to $\sum_{k\in\Ks} a^{\dl}_{n,k} (W_0 w_k^{\dl}) p^{\dl}_k\leq q^{\dl}_{\max, n}$, $\forall n\in\Ns$, and $(W_0 w_k^{\ul})p_k^{\ul}\leq p_{\max,k}^{\ul}$, $\forall k\in\Ks$. This means that the total transmit power per transmitter, computed as PSD multiplied by the total number of occupied RBs, is constrained by the predefined maximum power budget. Note that $\diag(\vw)\vp$ and $\diag(\vp)\vw$ are interchangeable. Moreover, for any fixed $\hat{\vp}$ or $\hat{\vw}$, the function $g_2$ over the joint variable $(\vw, \vp)$ can be written as $g_{2,\hat{\vw}}(\vp): \RN^{2K}\to\RN$ or $g_{2,\hat{\vp}}(\vw): \RN^{2K}\to\RN$.

%
%

\subsection{Link Gain Coupling Matrix}\label{subsec:linkgaincoupling}
%
The interference coupling between users (as shown in Fig. \ref{fig:ULDLSplit}) is characterized by a link gain coupling matrix. To define this matrix, we define three channel gain matrices $\mH_0\in\RP^{N\times K}$, $\mH_1\in\RP^{N\times N}$ and $\mH_2\in\RP^{K\times K}$ to indicate BS-to-UE, BS-to-BS, and UE-to-UE channel gain, respectively. 
The link gain coupling matrix between the $2K$ transmission links (UL and DL) is then defined to be
\begin{align}
\mV &  \coloneqq  \left[
\begin{array}{cc}
\mV^{\ulul} & \mV^{\dlul}\\
\mV^{\uldl} & \mV^{\dldl}
\end{array}
\right]
\label{eqn:Vmat_1}
\\
& =
\left[
\begin{array}{cc}
{\mA^{\ul}}^T\mH_0 & {\mA^{\ul}}^T\mH_1{\mA^{\dl}}\\
\mH_2 & \mH_0^T\mA^{\dl}
\end{array}
\right].
\label{eqn:Vmat_2}
\end{align}
The matrices $\mV^{\text{X}\leftarrow\text{Y}}:=\left(v^{\text{X}\leftarrow\text{Y}}_{k,j}\right)\in\RP^{K\times K}$, $\text{X},\text{Y}\in\{\ul, \dl\}$, determine the cross-link couplings. For example, $v^{\dlul}_{k,j}$  denotes the channel gain coupling between the transmitter of the downlink to UE $j$ and the receiver of the uplink from UE $k$ as shown in Fig. \ref{fig:ULDLSplit}. Note that $\mV^{\ulul}, \mV^{\dlul}$ and $\mV^{\dldl}$ are in general not symmetric, while $\mV^{\uldl}$ is symmetric.

We assume that each base station employs an OFDM-based scheme for resource allocation to schedule its users on orthogonal resources. As a result, there is no intra-cell interference and the interference coupling is completely described by the modified link gain matrix $\mVt$, which is defined by \eqref{eqn:Vmat_1} with $\mV^{\text{X}\leftarrow\text{Y}}$ replaced by $\mVt^{\text{X}\leftarrow\text{Y}}:=\left(\tilde{v}^{\text{X}\leftarrow\text{Y}}_{k,j}\right)$ where
\begin{equation}
\tilde{v}^{\text{X}\leftarrow\text{Y}}_{k,l}   \coloneqq 
\begin{cases}
v^{\text{X}\leftarrow\text{Y}}_{k,l} & \mbox{ if } b_l^{\text{Y}} \neq b_k^{\text{X}}\\
0 & \mbox{ o.w.}
\end{cases}.
\label{eqn:modifyV}
\end{equation} 
Here and hereafter, $b_k^{\text{X}}$, $\text{X} \in \{\ul, \dl\}$ denotes the serving BS of UE $k$ in UL or DL.

\subsection{Models of SINR and Rate}

\begin{figure}[t]
\centering
\includegraphics[width=.8\columnwidth]{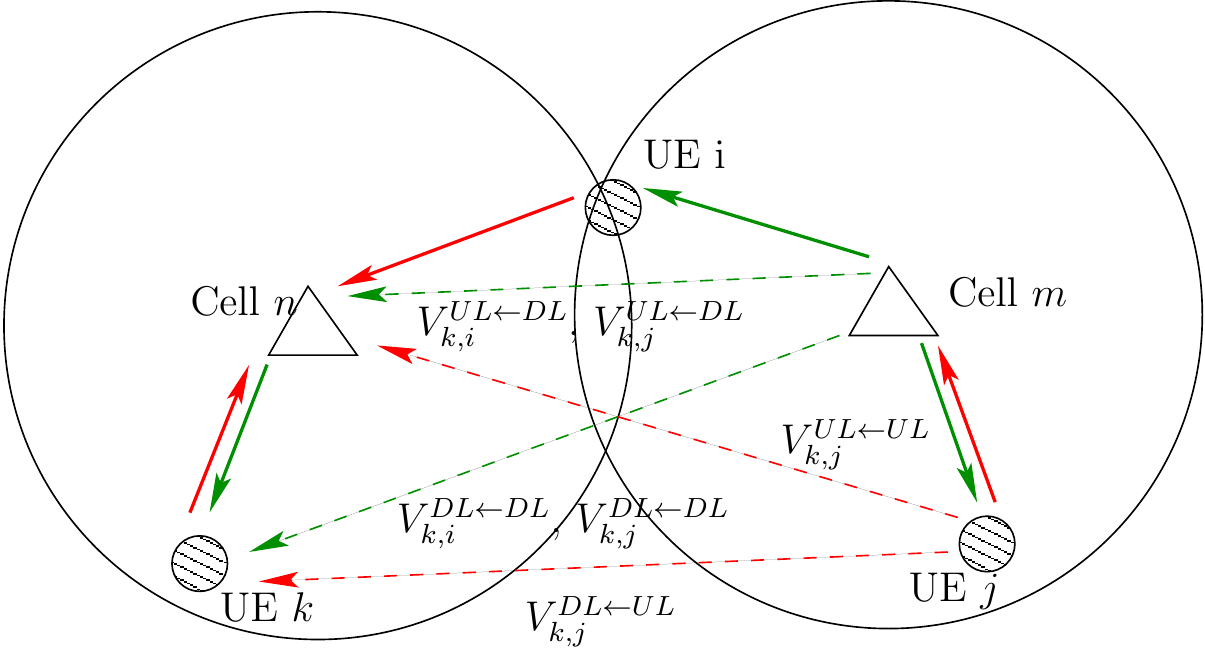}
\caption{Inter-cell interference coupling on the per-user basis. UE $i$ is associated to $n$ in UL and to cell $m$ in DL.}
\label{fig:ULDLSplit}
\end{figure}

To capture the dynamic inter-cell interference  in OFDM systems, it is reasonable to assume that the inter-cell interference increases as the fraction of the allocated RBs at the interfering BSs increases as well.
We interpret $\vw$ as the probability of generating interference from the transmitter of a link to the receiver of the other link (on any RB)\cite{mogensen2007lte}. More precisely, we assume that the DL and UL SINR 
 per RB of UE $k$ is given by (respectively)
\begin{align*}
\sinr_k^{\dl} &  \coloneqq  \frac{p_k^{\dl} v_{k,k}^{\dldl}}{
\sum\limits_{i\in\Ks}\tilde{v}^{\dldl}_{k,i}w_i^{\dl}p_i^{\dl} +
\sum\limits_{j\in\Ks}\tilde{v}_{k,j}^{\uldl}w_j^{\ul}p_j^{\ul}+\sigma^2
}
\\
\sinr_k^{\ul} &  \coloneqq \frac{p_k^{\ul} v_{k,k}^{\ulul}}{
\sum\limits_{i\in\Ks}\tilde{v}_{k,i}^{\dlul}w_i^{\dl}p_i^{\dl} +
\sum\limits_{j\in\Ks}\tilde{v}^{\ulul}_{k,j}w_j^{\ul}p_j^{\ul}+\sigma^2
}
\end{align*}
where $\sigma^2>0$ denotes the background-noise power spectral density, which is assumed to be the same for all receivers. 
Let us define $\ve{\sigma}  \coloneqq  \sigma^2\ve{1}_{2K}$, and collect the uplink and downlink SINR in a vector $\ve{\sinr} \coloneqq  [\sinr_1^{\ul}; \ldots; \sinr_K^{\ul}; \sinr_1^{\dl};\ldots; \sinr_K^{\dl}]\in\RP^{2K}$. Using \eqref{eqn:Vmat_1}, \eqref{eqn:Vmat_2}, and \eqref{eqn:modifyV}, the above expressions of SINR can be written in a general form
\begin{equation}
\sinr_l(\vp, \vw)  \coloneqq \frac{p_l}{\left[\mD^{-1}\left(\mVt\diag\{\vp\}\vw +\ve{\sigma}\right)\right]_l}, \ l\in\overline{\Ks},
\label{eqn:virtualSINR}
\end{equation}
where $\mD  \coloneqq \diag\{v_{1,1}^{\ulul}, \ldots, v_{K,K}^{\ulul}, v_{1,1}^{\dldl}, v_{K,K}^{\dldl}\}\in\RN^{2K}$ is a diagonal matrix. For $l=1, \ldots, K$, \eqref{eqn:virtualSINR} is equal to the UL SINR, while the DL SINR is given by \eqref{eqn:virtualSINR} for $l = K+1, \ldots, 2K$.

We further assume that the spectral efficiency (bit rate per RB) of the virtual UEs (includes both UL and DL transmission)  is a strictly increasing function of the SINR given by 
\begin{equation}
r_l (\vp, \vw)  \coloneqq  B\log_2(1 + \sinr_l(\vp, \vw)), \ l\in\overline{\Ks},
\label{eqn:rate}
\end{equation}
where $B$ denotes the effective bandwidth per RB.

Given the per-UE uplink and downlink traffic demands (bit rate) $\vd \coloneqq [d_1^{\ul},\ldots, d_K^{\ul}, d_1^{\dl}, \ldots, d_K^{\dl}]^T\in\RP^{2K}$, it follows from \eqref{eqn:rate} that the traffic demands are satisfied if and only if (note that $w_l\cdot W_0$ is equal to the number of RBs used by link $l$)
\begin{equation}
w_l\geq \frac{d_l}{W_0 r_l(\vp, \vw)}, \ l\in\overline{\Ks}.
\label{eqn:rateConstraint}
\end{equation}

\begin{Rem}[Full Overlap or Partial Overlap]\label{Rem:pverlapfactor}
The SINR modeled in \eqref{eqn:virtualSINR} is based on the strategy that each UL or DL transmission is allocated a number of RBs in a joint frequency band for both UL and DL, regardless of the location of the band. However, this may result in a full overlap of frequency bands used by UL and DL transmissions leading to high probability of inter-link interference. A more reasonable strategy is to allow only partial overlap, as shown in Fig. \ref{fig:DynamicFDDTDD_2}, where the DL is preferred to allocated at the head of the band while the UL at the tail of the band, or vice versa. 
In this case, the inter-link interference only exists on the overlapping band, and the above-presented model overestimates the probability of receiving inter-link interference. A more accurate readjustment is to multiply the term of inter-link interference with an additional overlap factor. Some possible methods to define the overlap factor are given in Appendix \ref{subsec:Overlap}. In the remainder of this paper, the analysis and algorithms are still presented with the interference model in \eqref{eqn:virtualSINR} for the simplicity of the form. However, without loss of generality, we can easily adjust the model by introducing the overlap factor into the coupling matrix $\mVt$. 
\end{Rem}

\subsection{Link Association Policies}\label{subsec:LinkAssociation}
Assume that there are a finite set of link association policies $\menge{\Pi} \coloneqq \{\pi_m: m=1,\ldots, M\}$ implemented in the network, which can be dynamically selected by an operator. Each policy defines the BS-UE assignment matrices $\mA^{\ul}(\pi_m)$ and $\mA^{\dl}(\pi_m)$, and further defines the link gain coupling matrix $\mVt(\pi_m)$ and link gain matrix $\mD(\pi_m)$ in \eqref{eqn:virtualSINR}.

As examples, in the following we list one conventional UL/DL coupled user association policy and two types of decoupled UL/DL link association policies, respectively.
\\
(1) {\bf CoUD}:  Conventional coupled UL/DL user association based on reference signal received power (RSRP) in DL is given by 
\begin{equation}
b_k^{\ul}=b_k^{\dl} = \argmax_{n\in\Ns} \rsrp_{n,k}, \ \forall k\in\Ks.
\label{eqn:crit_CoUD}
\end{equation}
(2) {\bf DeUD\_O}: Decoupled UL/DL link association assisted with cell selection offset \cite{europe2008range}. A cell selection offset is added to the reference signals of the small cells to increase their coverage in UL in order to offload some traffic from the macro cell. This can be formalized as follows
\begin{equation}
b_k^{\text{X}} =
\argmax_{n\in\Ns} \rsrp_{n,k} + \text{offset}^{\text{X}}_n, \ \forall k\in\Ks, \text{X} \in \{\ul, \dl\}
\label{eqn:crit_DeUD_O}
\end{equation}   
where $\text{offset}^{\text{X}}_n>0$ (in dB) if $\text{X}=\ul$ and $n$ is a small cell BS with low transmit power; otherwise the offset is set to zero if $\text{X}=\dl$ or $n$ is a macro cell BS.
\\
(3) {\bf DeUD\_P}:  Decoupled UL/DL link association based on DL received power and UL pathloss respectively \cite{elshaer2014downlink}, where the association criteria in DL and UL are given by (respectively) 
\begin{align}
b_k^{\dl} & = \argmax_{n\in\Ns} \rsrp_{n,k}, \\
b_k^{\ul} & = \argmax_{n\in\Ns} \text{PL}_{n,k}, \ \forall k\in\Ks,
\label{eqn:crit_DeUD_P}
\end{align} 
where $\text{PL}_{n,k}$ denotes the pathloss estimate between BS $n$ and UE $k$.

Note that in \eqref{eqn:crit_DeUD_O}, by setting $\text{offset}^{\text{X}}_n = 0$ for all $n\in\Ns$ and $\text{X}=\ul$, the association policy is equivalent to CoUD. And, by setting the offset (in dB) of the small cell BS in UL as the difference between the transmit power (in dBm) of the macro cell BS and the small cell BS,  DeUD\_O is equivalent to DeUD\_P.

  

\section{Problem Formulation}\label{ProbFormulation}
To achieve the service-centric network fairness,  we define the objective utility $\lambda$ to be the {\it minimum level of QoS satisfaction} among all links, where the level of QoS satisfaction is equal to the ratio of the per-link feasible transmission rate to the required traffic demand. So we have
\begin{equation}
\lambda := \min_{l\in\overline{\Ks}}  \frac{W_0 w_l r_l(\vp, \vw)}{d_l},
\label{eqn:utility_lambda}
\end{equation}
where $r_l(\vp, \vw)$ is given by \eqref{eqn:rate}.

Given a certain link association policy $\pi'$ and its corresponding UL(DL) assignment matrice $\mA^{\ul}(\pi')$( $\mA^{\dl}(\pi')$), coupling matrix $\tilde{\mV}(\pi')$, and link gain matrix $\mD(\pi')$, the objective is to maximize the utility $\lambda$ over the joint space of loads and powers subject to the constraints on the maximum per-cell load \eqref{eqn:loadConstraint} and the maximum per-transmitter power \eqref{eqn:powConstraint}. 
Moreover, if the optimized utility satisfies $\lambda\geq 1$, then the vector of link-specific traffic demands $\vd$ is feasible; otherwise, the traffic demand cannot be satisfied for every service link.
Formally, the problem of interest in this paper can be stated as follows.

\begin{problem}
\begin{subequations}
\label{eqn:prob_1}
\begin{align}
\max_{\vw\in\RN^{2K}, \vp\in\RN^{2K}} \ & \lambda  \label{eqn:prob_1_objective}\\
\mbox{subject to } &  \vw \geq \lambda \vf(\vp, \vw) \label{eqn:prob_1_utility}\\
 &  f_l(\vp, \vw)  \coloneqq \frac{d_l}{W_0 r_l(\vp, \vw)}, \forall l\in\overline{\Ks} 
\label{eqn:prob_1_def_load}\\
& \eqref{eqn:loadConstraint}, \eqref{eqn:powConstraint}, \label{eqn:prob_1_constraint}
\end{align}
\end{subequations}
where the vector function $\vf:\RN^{2K}\to\RP^{2K}$ in \eqref{eqn:prob_1_utility} is a collection of $f_l$ defined in \eqref{eqn:prob_1_def_load}, i.e., $\vf \coloneqq [f_1,\ldots, f_{2K}]^T$. 
\label{prob:jointOpt}
\end{problem}
The utility $\lambda$ depends on the joint variable $(\vw, \vp)\in\RN^{2K}\times\RN^{2K}$ in an inextricably intertwined manner, which is due to the nonlinear power and resource coupling relationship between links. 
%
We decompose Problem \ref{prob:jointOpt} into two subproblems in Problem \ref{prob:prob_2} by alternately optimizing over $\vw$ or $\vp$, and provide computationally efficient locally optimal solution to Problem \ref{prob:jointOpt}, based on the optimal solution to each of the subproblems. 

\begin{problem}\mbox{}
\begin{subproblem}\label{Prob:2a}
 Given fixed $\vp'\in\RN^{2K}$, find $\vw' := \vw'(\vp')$ such that
\begin{subequations}
\label{eqn:prob_2a}
\begin{align}
\vw' = \argmax_{\substack{\vw\in\RN^{2K}}} \ & \lambda \label{eqn:prob_2a_a}\\
\mbox{subject to } & \vw \geq \lambda \vf_{\vp'}(\vw)  \label{eqn:prob_2a_b}\\
& g_1(\vw)\leq 1, g_{2,  \vp'}(\vw)\leq 1, \label{eqn:prob_2a_c}
\end{align}
\end{subequations}
where $\vf_{\vp'}$, $g_1 $, and $g_{2, \vp'}$ are obtained by replacing $\vp$ with $\vp'$ in \eqref{eqn:prob_1_def_load}, \eqref{eqn:loadConstraint} and \eqref{eqn:powConstraint}, respectively. 
\end{subproblem}
\begin{subproblem}\label{Prob:2b}
 Given fixed $\vw'\in\RN^{2K}$ satisfying $g_1(\vw')\leq 1$, find $\vp' := \vp'(\vw')$ such that
\begin{subequations}
\label{eqn:prob_2b}
\begin{align}
\vp' = \argmax_{\substack{\vp\in\RN^{2K}}}  \ & \lambda \label{eqn:prob_2b_a}\\
 \mbox{subject to } & \vw'  \geq \lambda \vf_{\vw'}(\vp)\label{eqn:prob_2b_b}\\
 &  g_{2,  \vw'}(\vp)\leq 1, \label{eqn:prob_2b_c}
\end{align}
\end{subequations}
where $\vf_{\vw'}$ and $g_{2, \vw'}$ are obtained by replacing $ \vw$ with $\vw'$ in \eqref{eqn:prob_1_def_load} and \eqref{eqn:powConstraint}, respectively. 
\label{prob:prob_2}
\end{subproblem}
%
%
\label{prob:subporblems} 
\end{problem}

Prob.\ref{Prob:2a} and Prob.\ref{Prob:2b} are formulated in such a way that a common desired utility $\lambda$ is maximized subject to the common load and power constraints. Thus, for a given link association policy $\pi'$, by sequentially solving Prob.\ref{Prob:2a} and Prob.\ref{Prob:2b}, we improve $\lambda$ in each step and achieve a local optimum of $\lambda$ with respect to $\pi'$. 


In Section \ref{sec:ResourceAllo} and \ref{sec:PowerControl} we provide the optimal solution to Prob.\ref{Prob:2a} and Prob.\ref{Prob:2b}, respectively. The joint algorithm is summarized in Section \ref{sec:algorithm}.

\section{Joint Uplink and Downlink Resource Allocation}\label{sec:ResourceAllo}
In this section we develop the algorithms for joint UL/DL bandwidth allocation. In Section \ref{subsec:BW_Allocation} we develop an algorithm for Prob.\ref{Prob:2a} in Prop. \ref{prop:solu_FP_load}. Since a solution $\vw$ to Prob.\ref{Prob:2a} must fulfill $\max\{g_1(\vw), g_{2, \vp'}(\vw)\}\leq 1$, 
some free resources may still be available, i.e., $g_1(\vw)<1$ and $g_{2, \vp'}(\vw) = 1$, even under optimal power allocation (in the sense of Prob. \ref{Prob:2a}). Therefore, an additional step involving power scaling and bandwidth updating is introduced in Prop. \ref{prop:fullload} in Section \ref{subsec:BW_allo_fullload}, to further improve the desired utility $\lambda$. Another case of $g_1(\vw)=1$ and $g_{2, \vp'}(\vw)\leq 1$ is discussed in Prop. \ref{eqn:prop_power_allo} in Section \ref{sec:PowerControl}.

\subsection{Algorithm for Bandwidth Allocation}\label{subsec:BW_Allocation}
The following lemma proves a key property of the vector function $\vf_{\vp'}$, which is necessary to solve  Prob. \ref{Prob:2a}.

%
\begin{lemma}
Given a fixed power vector $\vp'$, the function $\vf_{\vp'}:\RN^{2K}\to\RP^{2K}$ defined in Prob. \ref{prob:jointOpt} is a standard interference function. 
\label{lem:StandardInterference}
\end{lemma}
The definition and some selected properties of standard interference function (SIF) are provided in Appendix \ref{subset:SIF}. The proof of Lemma \ref{lem:StandardInterference} following the proof of \cite[Ex. 2]{Cavalcante14} is provided in Appendix \ref{subsec:proofSIF}. 

We further prove the following theorem.
\begin{theorem}
Suppose
\begin{itemize}
\item $g(\vx):\RP^k\to\RP$ is monotonic, and homogeneous of degree $1$ (i.e., $g(\alpha\vx) = \alpha g(\vx)$ for all $\alpha>0$),  
\item $\vf(\vx):\RN^k\to\RP^k$ is a SIF.
\end{itemize}
Then, for each $\theta>0$ there is exactly one eigenvector $\vx'\in\RP^k$ and associate eigenvalue $\rho'$ of $\vf$ such that $\rho'\vx' = \vf(\vx')$ and $g(\vx')=\theta$. The repeated iteration
\begin{equation}
\vx^{(t+1)} =\frac{\theta \vf(\vx^{(t)})}{g\circ\vf(\vx^{(t)})}, \ t\in\NN,
\label{eqn:FP_scaling_mapping}
\end{equation}
converges to the unique vector $\vx'$, which is called the fixed point of $\vf$. The associate eigenvalue is $\rho'=g\circ\vf(\vx')/\theta$.
\label{the:MSS_FP}
\end{theorem}  
The proof of Theorem \ref{the:MSS_FP} is referred to Appendix \ref{subsec:proofThe_1}. It is a direct extension of the proof of \cite[Th. 3.2]{nuzman2007contraction}, where $g$ was defined as any monotonic norm $\|\cdot\|$, while we define three properties monotonicity, homogeneity and positivity on $\RP^k$. Note that the function in \eqref{eqn:FP_scaling_mapping} $\ve{\psi} \coloneqq \theta \vf/g\circ\vf: \RN^k\to\RP^{k}$  is non-monotonic, while it preserves the property of scalability of the mapping $\vf$.

Using Lemma \ref{lem:StandardInterference} and Theorem \ref{the:MSS_FP}, we prove the following proposition, which gives rise to an algorithmic solution to Prob.\ref{Prob:2a}.
\begin{proposition}
Given a fixed $\vp'\in\RN^{2K}$, let the set of solutions to Prob.\ref{Prob:2a} be denoted by $\Fs_{\vw}(\vp')$. 
There exists one $\vw'\in\Fs_{\vw}(\vp')$ such that $\vw'\leq \vw$ for all $\vw\in\Fs_{\vw}(\vp')$. Moreover, $\vw'$ is an eigenvector of $\vf_{\vp'}$ satisfying $\max\{g_1(\vw'), g_{2, \vp'}(\vw')\} = 1$ and
can be obtained by performing the following fixed point iteration:
\begin{subequations}
\label{eqn:FP_converge}
\begin{align}
\vw^{(t+1)} &= \frac{\vf_{\vp'}(\vw^{(t)})}{g_{\vp'}\circ \vf_{\vp'}(\vw^{(t)})}, \ t\in \NN, \label{eqn:FP_converge_1}\\
\mbox{where } g_{\vp'}(\vw)& \coloneqq  \max\{g_1(\vw), g_{2, \vp'}(\vw)\}.
\label{eqn:FP_converge_2}
\end{align}
\end{subequations}
The iteration in \eqref{eqn:FP_converge} converges to $\vw'$, and $\lambda_{\vp'} = 1/g_{\vp'}\circ \vf_{\vp'}(\vw')$. 
\label{prop:solu_FP_load} 
\end{proposition}
\begin{proof}
In the following part of this proof, for simplicity of notation, we omit the dependency on $\vp'$, and denote $\vf \coloneqq \vf_{\vp'}$, $g \coloneqq g_{\vp'}$ and $\lambda  \coloneqq  \lambda_{\vp'}$.

It is obvious that $g$ defined in \eqref{eqn:FP_converge_2} is positive and homogeneous of degree 1 on $\RP^{2K}$. By virtue of Theorem \ref{the:MSS_FP} and Lemma \ref{lem:StandardInterference}, we have that  for $\theta = 1$, there exist a unique fixed point $\vw' = \lambda'\vf(\vw')$ such that $g(\vw') = 1$, where $\lambda'$ can be computed with iteration \eqref{eqn:FP_converge_1}.

Then we show that there exists no $\lambda''> \lambda'$  to satisfy $\vw'' \geq \lambda''\vf(\vw'')$ and  $g(\vw'')\leq 1$. We proceed by contradiction. Suppose that there exists a $\lambda''> \lambda'$ to satisfy $\vw'' \geq \lambda''\vf(\vw'')$ such that  $g(\vw'')\leq 1$. 
Let us define a function $\vf' \coloneqq  \lambda'\vf$. Because $\vf$ is a SIF, $\vf'$ is also a SIF. We then have $\vf'(\vw'')= \lambda'\vf(\vw'')<\lambda''\vf(\vw'') \leq \vw''$, i.e., $\vw''$ is a feasible point with respect to the SIF $\vf'$. Thus, the sequence starting from $\vw''$ decreases monotonically to $\vw'$ (by using the third property of SIF stated in Lemma \ref{lem:Prop_SIF}).
Then we have $\vw'\leq \vf'(\vw'')<\vw''$. Since $g(\vw)$ is monotone increasing on $\RN^{2K}$, we have $g(\vw'')>g(\vw')=1$, which contradicts the earlier statement $g(\vw'')\leq 1$.

Knowing that $\lambda'$ is the maximum feasible utility, now we show that for all $\vw\in\Fs_{\vw}(\vp')$ satisfying $\vw\geq \lambda'\vf(\vw) =\vf'(\vw)$, we have $\vw'\leq \vw$. Because $\vf'$ is also a SIF,  $\vw\geq \vf'(\vw)$ implies that the sequence $\vw$ decreases monotonically to $\vw'$ satisfying $\vw'=\vf'(\vw')=\lambda'\vf(\vw')$. Thus., $\vw'\leq\vw$.
\end{proof}

\subsection{Optimization to Achieve Maximum Load}\label{subsec:BW_allo_fullload}
As aforementioned, Prop.\ref{prop:solu_FP_load} provides an algorithm that converges to the optimal solution to Prob.\ref{Prob:2a}. Let $\vw'$ be this solution. Since $\max\{g_1(\vw'), g_{2, \vp'}(\vw')\} = 1$, it is possible that $g_{2, \vp'}(\vw')=1$, while $g_1(\vw')<1$, i.e., the maximum power per transmitter is satisfied with equality, while free resources are still available. In this case, we propose an additional step to further optimize $\lambda$ by iteratively scaling down the fixed power vector $\vp'$, until $g_1(\vw')=1$ is achieved.
\begin{proposition}
Let $\vw'\in\RN^{2K}$ be the solution to Prob.\ref{Prob:2a} and suppose that $g_{2, \vp'}(\vw')=1$ and $g_{1}(\vw')<1$. Starting from $\vp^{(0)}=\vp'$ and $\vw^{(0)}=\vw'$, by iteratively performing the following two steps:
\begin{itemize}
\item[(1)] scaling down $\vp$ by
\begin{equation}
\vp^{(t+1)} = g_1(\vw^{(t)})\cdot\vp^{(t)},
\label{eqn:FP_scalePow}
\end{equation}
\item[(2)] updating $\vw^{(t+1)}$, as the unique fixed point of iteration \eqref{eqn:FP_converge}, with updated $\vp' = \vp^{(t+1)}$,
\end{itemize}
the sequence of utility $\lambda$ is monotone increasing, until the maximum load constraint $g_1(\vw) = 1$ is satisfied.
\label{prop:fullload}
\end{proposition}
The proof of Prop. \ref{prop:fullload} is provided in Appendix \ref{subsec:proofProp_2}.

The optimization step provided in Prop. \ref{prop:fullload} further improves our desired utility $\lambda$ if the solution to Prob.\ref{Prob:2a} $\vw'$ satisfies $g_{2, \vp'}(\vw')=1$ and $g_{1}(\vw')<1$. Now assume the algorithm defined in Prop. \ref{prop:fullload} converges to $(\vp^{\star}, \vw^{\star})$. 
Then, in addition to the full utilization of resources in the sense that $g_{1}(\vw^{\star}) = 1$, we have $g_{2}(\vp^{\star}, \vw^{\star}) \leq 1 = g_{2, \vp'}(\vw')$, which means that the allocation obtained under Prop. \ref{prop:solu_FP_load} is more power efficient than that of Prop. \ref{prop:solu_FP_load}.

\begin{Rem}\label{Rem:LoadFull}
It is worth mentioning that Ho \cite{ho2014power} formulates a power minimization problem, based on the cell-specific load and power coupling in the DL,  and concludes that if the minimum required rate is feasible, then the optimal solution to the power minimization problem satisfies that the system is fully loaded \cite[Th. 1]{ho2014power}. In this paper, we formulate a utility maximization problem, based on the link-specific bandwidth and power coupling framework in joint UL/DL, with per-cell load and per-transmitter power constraints, and conclude that if some minimum utility is feasible with cell load lower than one, we can scale down the power vector using the algorithm presented in Prop. \ref{prop:fullload}, to further increase the desired utility, until the per-cell load constraint holds with equality.
\end{Rem}

\section{Joint Uplink and Downlink Power Control}\label{sec:PowerControl}
Now let us consider the problem of power control. In this section, we first present the optimal solution to Prob.\ref{Prob:2b} introduced in Section \ref{subsec:PerLinkPowerControl}. Then, in Section \ref{subsec:percell_powercontrol} and \ref{subsec:efficient_powercontrol}, we further examine two alternative algorithms  for cell-specific power control and energy efficient power control, respectively.
\subsection{Algorithm for Link-Specific Power Control}\label{subsec:PerLinkPowerControl}
Let us first consider Prob.\ref{Prob:2b}. Given some fixed $\vw'\in[0,1]^{2K}$, we first rewrite the rate constraints in \eqref{eqn:prob_2b_b}. For $\vp\in\RP^{2K}$, we have
\begin{equation}
\vw' \geq \lambda \vf_{\vw'}(\vp) \Leftrightarrow p_l\geq \lambda \frac{p_l f_{\vw',l}(\vp)}{w'_l} \mbox{ for } l\in\overline{\Ks}.
\label{eqn:reformulate_w}
\end{equation}
We further define the following vector function using \eqref{eqn:reformulate_w}
\begin{align}
\tilde{\vf}_{\vw'}: &\RP^{2K}\to\RP^{2K}:\vp\mapsto\left[\tilde{f}_{\vw',1}(\vp), \ldots, \tilde{f}_{\vw',2K}(\vp)\right]^T\nonumber\\
\mbox{where } & \tilde{f}_{\vw',l}(\vp) \coloneqq  \dfrac{p_l}{w'_l} f_{\vw',l}(\vp), \ l\in\overline{\Ks}.
\label{eqn:def_funct_w}
\end{align}
Note that the domain of $\tilde{\vf}_{\vw'}$ defined in \eqref{eqn:def_funct_w} is the positive orthant $\RP^{2K}$. To extend it to the non-negative orthant $\RN^{2K}$, we define the following extension for each $l \in\overline{\Ks}$:
\begin{align}
\vf'_{\vw',l}(\vp) &  \coloneqq  
\begin{cases}
\tilde{\vf}_{\vw',l}, & \mbox{ if }  p_l\neq 0\\
\dfrac{d_l\ln 2}{W_0 B w'_l} \Is_{\vw',l}(\vp) & \mbox{ o.w. }
\end{cases},
\label{eqn:ext_funct_w}
\\
\mbox{where }  & \Is_{\vw',l}(\vp): = \left[\mD^{-1}\left(\mVt\diag\{\vw'\}\vp+\ve{\sigma}\right)\right]_l.
\label{eqn:Funct_Interf_w}
\end{align}
The domain extension is derived by leveraging the linear approximation $\log_2(1+x)\approx x/\ln 2$ for $x\to 0$. As shown in \eqref{eqn:ext_funct_w}, this approximation is only used for $p_l= 0$ (which further leads to $\sinr_l = 0$), otherwise if $p_l\neq 0$, the nonlinear closed-form of $\tilde{\vf}_{\vw',l}$ \eqref{eqn:def_funct_w} is used.

With \eqref{eqn:reformulate_w}, \eqref{eqn:ext_funct_w}, and \eqref{eqn:Funct_Interf_w},  Prob.\ref{Prob:2b} is rewritten as
\begin{equation}
\max_{\substack{\vp\in\RN^{2K}}} \lambda,\mbox{ s.t. } \vp \geq \lambda \vf'_{\vw'}(\vp), \ 
 g_{2, \vw'}(\vp)\leq 1 
\label{eqn:prob_2b_p}
\end{equation}
The following lemma shows that $\vf'_{\vw'}$ has the same key property as $\vf_{\vp'}$, which is shown for $\vf_{\vp'}$ in Lemma \ref{lem:StandardInterference}.
\begin{lemma}
The vector function $\vf'_{\vw'}:\RN^{2K}\to\RP^{2K}$ defined in \eqref{eqn:ext_funct_w} is SIF. 
\label{lem:SIF_Pow_FixedW} 
\end{lemma}
\begin{proof}
The proof follows directly from the previous results in \cite[Prop. 1]{cavalcante2014power},
where a cell-specific utility function over the cell-specific power vector in DL is shown to be positive concave, and thus a SIF \cite[Prop. 1]{Cavalcante14}. It is easy to see that our defined link-specific function $\vf'_{\pi', \vw'}$  shares the same form with the cell-specific function introduced in \cite[Prop. 1]{cavalcante2014power}. Thus, we omit the details here and conclude that it is also a SIF.
\end{proof}

Note that in the expression of per-transmitter power constraint \eqref{eqn:powConstraint}, the term 
 $\diag(\vw)\vp$ and $\diag(\vp)\vw$ are interchangeable. With some fixed $\vw'$, the function $g_{2, \vw'}$ defined in \eqref{eqn:prob_2b_p} is positive and homogeneous of degree $1$ on $\RP^{2K}$. Thus, by leveraging Lemma \ref{lem:SIF_Pow_FixedW} and Theorem \ref{the:MSS_FP}, we can argue along similar lines as in Prop. \ref{prop:solu_FP_load} to conclude the following: starting from an arbitrary $\vp^{(1)}\in\RN^{2K}$, the following fixed point iteration
\begin{equation}
\vp^{(t+1)} = \frac{\vf'_{\vw'}(\vp^{(t)})}{g_{2,\vw'}\circ \vf'_{\vw'}(\vp^{(t)})}, \ t\in \NN \label{eqn:FP_converge_fixp}
\end{equation}
converges to the solution of Prob.\ref{Prob:2b}, denoted by $\vp''$. And the utility $\lambda_{\vp''}$ corresponding to $\vp''$ is given by $\lambda_{\vp''} = 1/g_{2, \vw'}\circ \vf'_{\vw'}(\vp'')$.

Using \eqref{eqn:FP_converge_fixp}, we can iteratively approach arbitrarily close to solution to Prob.\ref{Prob:2b} given fixed $\vw'$ as the solution to Prob.\ref{Prob:2a}. However, for joint optimization over $(\vw, \vp)$, we are interested in  whether or not this solution further improves the desired utility derived from the solution to Prob.\ref{Prob:2a}. We present the relationship between $\lambda'':=\lambda_{\vp''}$ and $\lambda':=\lambda_{\vp'}$ in Prop. \ref{eqn:prop_power_allo}.

\begin{proposition}
For some fixed $\vp'$, let $\vw'\in\RP^{2K}$ be the solution to Prob.\ref{Prob:2a} and $\lambda'$ the corresponding utility. Moreover, given $\vw'$, let $\vp''\in\RP^{2K}$ be the solution to Prob.\ref{Prob:2b} and $\lambda''$ the corresponding utility. Then, $\lambda'$ and $\lambda''$ are related as follows.
\begin{itemize}
\item If $g_{2, \vp'}(\vw') = 1$, then, we have $\lambda''= \lambda'$ and  $\vp'' = \vp'$
\item If $g_{2, \vp'}(\vw') < 1$, then, we have $ \lambda''>\lambda'$
\end{itemize}
\label{eqn:prop_power_allo}
\end{proposition}
The proof of Prop. \ref{eqn:prop_power_allo} can be found in Appendix \ref{subsec:proofProp_3}. 

Prop. \ref{eqn:prop_power_allo} implies that given the solution $(\vw', \vp')$ derived from the bandwidth updating step (Prop. \ref{prop:solu_FP_load}) or the power scaling step (Prop. \ref{prop:fullload}),  with fixed $\vw'$ at hand,  solving Prob.\ref{Prob:2b} (by performing \eqref{eqn:FP_converge_fixp} ) can further improve the desired utility only if $g_{2,\vp'}(\vw') < 1$; otherwise if $g_{2,\vp'}(\vw') = 1$ the solution to Prob.\ref{Prob:2b} with respect to $\vw'$ is equivalent to $\vp'$.


%
\begin{Rem}
In this section, we rewrite the rate constraints $\vw'\geq \lambda \vf_{\vp}(\vw')$ in Prob. \ref{Prob:2b} into a system of  nonlinear inequalities $\vp\geq \lambda \vf_{\vw'}'(\vp)$ as shown in \eqref{eqn:reformulate_w}-\eqref{eqn:Funct_Interf_w}. Hence both the fixed point iterations in \eqref{eqn:FP_converge} and \eqref{eqn:FP_converge_fixp} (to solve Prob. \ref{Prob:2a} and Prob. \ref{Prob:2b}, respectively) converge to the solutions that maximize the same $\lambda$ defined in Prob. \ref{ProbFormulation}.
Note that if we treat the power control problems separately, as stated for instance in \cite{BocheSchubertEURASIP05}, the rate constraint $r_l(\vp,\vw')\geq \lambda d_l/(w_l'W_0)$ for all $l\in\overline{\Ks}$ can be directly translate into a SINR constraint by taking the exponential function of both sides.  We write \eqref{eqn:reformulate_w}  into a system of linear inequalities in powers:
\begin{equation}
p_l\geq \eta(\lambda) \vf^{''}_{\vw'}(\vp)\nonumber
\end{equation}
where $\eta(\lambda):=2^{\frac{\lambda d_l}{W_0Bw_l'}}-1$ is monotone increasing for any $\lambda\in\RN^{2K}$, and $f^{''}_{\vw'}:\RN^{2K}\to\RP^{2K}$ is of form of an affine transformation $\vp\mapsto\mD^{-1}\left(\mVt\diag(\vw')\vp+\ve{\sigma}\right)$. We can agree along similar lines as in Prop. \ref{prop:solu_FP_load} to maximize $\eta$ by performing the fixed point iteration $\vp = f^{''}_{\vw'}(\vp)/ (g_{2,\vw'}\circ f^{''}_{\vw'}(\vp))$ and thus indirectly maximize $\lambda$.
\end{Rem}

\subsection{Algorithm for Cell-Specific Power Control}\label{subsec:percell_powercontrol}
So far we have considered the case that the PSD $\vp$ can be specified per service link. In the practical system, however, in DL a transmitter determines constant cell-specific energy per resource element across all DL bandwidth and subframes until it needs to be updated \cite{3GPPTS36213}, while   
in UL a distinct transmission power can be assigned to each UE. Without loss of generality, the developed power control algorithm can be easily modified to meet this practical requirement. 
The objective is  to optimize the per-transmitter PSD as a collection of the per-UE UL and per-BS DL power vectors
 \begin{equation}
\vqt  \coloneqq  [\vp^{\ul};\vq^{\dl}]^T\in\RN^{K+N},
\label{eqn:power_collect}
\end{equation}
where $\vq^{\dl}\in\RN^N$ is the cell-specific PSD in DL, and the $n$th entry of $q_l^{\dl}$ denotes the PSD of all the DLs associated to cell $n$.
Since all DLs served by the same cell share the same PSD, we have
\begin{equation}
\vp^{\dl}={\mA^{\dl}}^T\vq^{\dl}.
\label{eqn:power_percell}
\end{equation}
%
%
The transformation between $\vp$ and $\vqt$ is then given by 
\begin{equation}
\vp = \mLa\vqt, \mbox{ with }
\mLa  \coloneqq  
\left[ \begin{array}{cc}
\mI_{K} & \ma{0}_{K\times N}\\
\ma{0}_{K\times K} & {\mA^{\dl}}^{T}
\end{array}
\right]. 
\label{eqn:transformationP}
\end{equation}

In the following, we collect the per-UE rate constraint in UL and per-cell sum rate constraint in DL depending on $\vqt$ in a set of $K+N$ nonlinear inequalities, where for $j\in\Ks$  the $j$th inequality implies the UL rate constraint for UE $j$, while for $j \in\overline{\Ns}  \coloneqq \{K+1, \ldots, K+N\}$, the $j$th  inequality implies the DL sum rate constraint for cell $n=j-K$.
\subsubsection{Per-UE Rate Constraint in Uplink}
Substituting \eqref{eqn:transformationP} into \eqref{eqn:virtualSINR}, SINR of UE $j$ in UL is simply given by
\begin{align}
\sinr_j (\vqt, \vw') & \coloneqq  \frac{\pt_j}{\Js_{\vw',j}(\vqt)}, \mbox{ for } j\in\Ks, \label{eqn:UL_SINR_pow}\\
\mbox{ where } \Js_{\vw',j}(\vqt) &  \coloneqq  \left[\mD^{-1}\left(\mVt\diag\{\vw'\}\mLa\vqt+\ve{\sigma}\right)\right]_j.
\label{eqn:interfernceFunct_pow}
\end{align}
Substituting \eqref{eqn:UL_SINR_pow} into \eqref{eqn:rate} and \eqref{eqn:rateConstraint}, the per-UE rate constraint in UL depending on $\vqt$ is given by
\begin{equation}
\pt_j\geq \frac{\pt_j}{w_j}\cdot \frac{d_j}{W_0 r_j(\vqt, \vw')} \eqqcolon \tilde{f}_{\vw',j}(\vqt), \mbox{ for }  j \in \Ks.
\label{eqn:UL_rateConstraint}
\end{equation}
\subsubsection{Per-Cell Sum Rate Constraint in Downlink}
Substituting \eqref{eqn:transformationP} into \eqref{eqn:virtualSINR}, the DL SINR of UE $k$ associated with cell $n$ (depending on $\vqt$) can be rewritten as:
\begin{equation}
\sinr_{n,l}^{\dl}(\vqt, \vw') \coloneqq \dfrac{\pt_{K+n}}{\Js_{\vw',l}(\vqt)}, \ \forall l\in\overline{\Ks}_n^{\dl},
\label{eqn:DL_SINR_pow}
\end{equation}
where $\Js_{\vw',l}(\vqt)$ is defined in \eqref{eqn:interfernceFunct_pow}, $\overline{\Ks}_n^{\dl}$ denotes the set of DL transmissions associated with cell $n$, and $\pt_{K+n}$ as the $(K+n)$th entry of $\vqt$ denotes the PSD in DL in cell $n$. 

The spectral efficiency of UE $k$ associated with cell $n$ in DL and denoted by $r^{\dl}_{n,l}(\vqt,\vw')$ is computed by substituting \eqref{eqn:DL_SINR_pow} into \eqref{eqn:rate}.
Then, using \eqref{eqn:rateConstraint}, the sum rate constraint per cell in DL (depending on $\vqt$) yields
\begin{align}
\nu_n' & = \sum_{l\in\overline{\Ks}_n^{\dl}}w'_l\geq \sum_{l\in\overline{\Ks}_n^{\dl}} 
\frac{d_l}{W_0 r_{n,l}^{\dl}(\vqt,\vw')},  \ \forall n\in\Ns \\
\Rightarrow \pt_j &\geq \frac{\pt_j}{\nu_{j-K}'}\sum_{l\in\overline{\Ks}_{j-K}^{\dl}} 
\frac{d_l}{W_0 r_{j-K,l}^{\dl}(\vqt,\vw')} \nonumber\\
 & \eqqcolon \tilde{f}_{\vw',j}(\vqt), \mbox{ for } j \in\overline{\Ns}
\label{eqn:DL_rateConstraint}
\end{align}
where $\nu_n'$ denotes fraction of the total allocated RBs of cell $n$ in DL,\; note that for $j\in\overline{\Ns}$, the $j$th entry of $\vqt$ is equal to the PSD of cell $n=j-K$ in DL.

Note that \eqref{eqn:DL_rateConstraint} defines the $j$th entry of function $\tilde{f}_{\vw',j}$ for $j = K+1, \ldots, K+N$, while for $j = 1,\ldots, K$, the expression of  $\tilde{f}_{\vw',j}$ is given in \eqref{eqn:UL_rateConstraint}.

\subsubsection{Joint Downlink Cell-Specific and Uplink UE-specific Power Control}

With \eqref{eqn:UL_rateConstraint} and \eqref{eqn:DL_rateConstraint} in hand, using the same techniques as shown in \eqref{eqn:reformulate_w}-\eqref{eqn:Funct_Interf_w}, the optimization problem is written as 
\begin{equation}
\max_{\substack{\vqt\in\RN^{K+N}}}  \lambda,\mbox{ s.t. } \vqt \geq \lambda \vfo_{\vw'}(\vqt)
, \ \overline{g}_{2, \vw'}(\vqt)\leq 1 
\label{eqn:FP_load_fixed_w_modi_3}
\end{equation}
where $\overline{g}_{2, \vw'}(\vqt)$ is obtained by substituting \eqref{eqn:transformationP} into \eqref{eqn:powConstraint}, and $\vfo_{\vw'}(\vqt)$ is given by
\begin{align}
& \sfo_{\vw',j}(\vqt) \coloneqq \nonumber \\
&
\begin{cases}
\tilde{f}_{\vw',j}(\vqt) & \mbox{ if } \pt_j \neq 0\\
\dfrac{d_l\ln 2}{W_0 B w'_j} \Js_{\vw',j}(\vqt) & \mbox{ if } \pt_j = 0, j\in\Ks\\
\sum\limits_{l\in\overline{\Ks}_{j-K}^{\dl}}\dfrac{d_l\ln2
}{W_0B\nu_{j-K}'} \Js_{\vw',l}(\vqt) & \mbox{ if } \pt_j = 0, j\in\overline{\Ns}\\
\end{cases}
\end{align}
Proceeding long similar lines as in Lemma \ref{lem:SIF_Pow_FixedW}, it is easy to show that $\vfo_{\vw'}:\RN^{K+N}\to\RP^{K+N}$ is SIF, while $\overline{g}_{2, \vw'}:\RP^{K+N}\to\RP$ is monotonic and homogeneous with degree $1$. Therefore, we can compute the solution to \eqref{eqn:FP_load_fixed_w_modi_3} by means of the fixed point iteration in \eqref{eqn:FP_converge_fixp}, and with $\vf'_{\vw'}(\vp)$ replaced by $\vfo_{\vw'}(\vqt)$.

\subsection{Algorithm for Energy Efficient Power Control}\label{subsec:efficient_powercontrol}
If the following assumption holds, the rate requirements are strictly feasible for all UL and DL transmissions.
\begin{assumption}
The solution to Prob. \ref{prob:subporblems} $(\vw^{\star},\vp^{\star})$ satisfies $\lambda^{\star}>1$.
\label{assume:feasible}
\end{assumption}
Under Assumption \ref{assume:feasible}, the problem of interest in the context of energy efficient networks is that, instead of consuming high energy to achieve $\lambda >1$, how to minimize the sum transmit power, such that the per-link rate constraint is just satisfied, i.e., $\lambda = 1$. 
The power minimization problem subjected to the rate and power constraints are defined in Problem \ref{prob:minPow}

\begin{problem}
\begin{equation}
\min_{\vp\in\RN^{2K}} \psi(\vp), \mbox{ s.t. }  \vp \geq \vf'_{\vw^{\star}}(\vp), \ g_{2, \vw^{\star}}(\vp)\leq 1
\label{eqn:prob_minPow}
\end{equation}
where $\psi:\RN^{2K}\to\RN$ can be any monotonic function (in each coordinate, i.e., $\psi(\vx)\geq \psi(\vy)$ if $x_i\geq y_i$ for each $i$) that is non-decreasing.
For example, by setting $\psi(\vp) = \|\diag\{\vw^{\star}\}\vp\|_1$, we aim at minimizing the sum transmit power  over all occupied RBs and all transmitters. 
\label{prob:minPow}
\end{problem}

Since $\vf'_{\vw^{\star}}$ is SIF, Prob. \ref{prob:minPow} is a classical power minimization problem introduced in \cite{Yates95b}, and we provide the solution in Prop. \eqref{prop:minPow}. We omit the proof because it follows directly from \cite[Thm. 2]{Yates95b}.
\begin{proposition}
Under Assumption \ref{assume:feasible}, the fixed point iteration
\begin{equation}
\vp^{(t+1)} = \vf'_{\vw^{\star}}\left(\vp^{(t)}\right), t\in\NN
\label{eqn:FP_converge_minP}
\end{equation}
converges to the optimum solution $\vp^{\star\star}$ to Prob. \ref{prob:minPow}.
\label{prop:minPow}
\end{proposition}

Note that without loss of generality, \eqref{eqn:prob_minPow} can be easily translated to the power minimization problem over $\vqt$  by substituting \eqref{eqn:transformationP} into \eqref{eqn:prob_minPow} and replacing $\vf'_{\vw^{\star}}$ with $\vfo_{\vw'}$. 
%
%

\section{Algorithm for Joint Optimization}\label{sec:algorithm}
Now we provide an algorithm for joint optimization of bandwidth allocation $\vw$ and power control $\vp$ per link, with respect to any fixed link association policy $\pi'\in\Pi$. Based on Prop. \ref{prop:solu_FP_load}, \ref{prop:fullload}, and \ref{eqn:prop_power_allo}, we can compute the locally optimum of $\left(\vw(\pi'), \vp(\pi')\right)$. 
%
In the following we explain in more detail the three main steps (S1, S2 and S3) of the algorithm. 
%
\vspace{4pt}
\\
{\it S1. Updating Bandwidth}
\vspace{3pt}
\\
The algorithm starts with optimizing the bandwidth allocation $\vw$, given an initial PSD $\vp'$.
Prop. \ref{prop:solu_FP_load} provides the optimal solution $\vw'$ in the sense of maximizing $\lambda$ for any fixed $ \vp'$. The algorithm converges to a solution $\vw'$, satisfying $\max\{g_{1}(\vw'), g_{2, \vp'}(\vw')\} = 1$, i.e., either $g_{1}(\vw') = 1$,  or $g_{2, \vp'}(\vw') =1$, or both. Therefore, it remains to consider the following three cases
\begin{itemize}
\item[(1)]  $g_{1}(\vw') < 1$ and $g_{2, \vp'}(\vw')=1$
\item[(2)]  $g_{1}(\vw') = 1$ and $g_{2, \vp'}(\vw')<1$
\item[(3)]  $g_{1}(\vw') = 1$ and $g_{2, \vp'}(\vw')=1$
\end{itemize}
Note that once the third condition is achieved, $(\vw',\vp')$ is a local optimum. In contrast,  in the first case and the second case the algorithm is designed to further improve the utility by proceeding with S2 and S3 (see Algorithm \ref{algo:JointOptimization}), respectively.
\vspace{4pt}
\\
{\it S2. Power Scaling to Achieve The Full Load Condition}
\vspace{3pt}
\\
The first condition leads to the power scaling step as described in Prop. \ref{prop:fullload}. At this step, power scaling \eqref{eqn:FP_scalePow} and bandwidth updating \eqref{eqn:FP_converge} are performed iteratively, until the solution $(\vp', \vw')$ converges and satisfies $g_{1}(\vw') = 1$ and $g_{2, \vp'}(\vw')\leq 1$.
\begin{itemize}
\item[(1)]
If $g_{2, \vp'}(\vw')= 1$, then $(\vp',\vw')$ is considered the local optimum. 
\item[(2)] If $g_{2, \vp'}(\vw')< 1$, then the algorithm moves to the power updating step S3. 
\end{itemize}
\vspace{4pt}
{\it S3. Updating Power Budget}
\vspace{3pt}
\\
As shown in Prop. \ref{eqn:prop_power_allo}, the power updating step improves the utility if $g_{2, \vp'}(\vw') <1$, where $(\vw',\vp')$ are derived from the bandwidth updating step S1. Therefore, the algorithm moves to S3 if either of the following conditions holds.
\begin{itemize}
\item[(1)] S1 returns $g_{1}(\vw') = 1$ and $g_{2, \vp'}(\vw')<1$, and the algorithm moves directly to S3.
\item[(2)] S1 returns $g_{1}(\vw') < 1$ and $g_{2, \vp'}(\vw')=1$, and the algorithm moves to S2. If S2 returns $g_{1}(\vw') =1$ and $g_{2, \vp'}(\vw')< 1$, then, algorithm further moves to S3.
\end{itemize}

\begin{Rem}[Selection of The Initial Point]
The initial point has in general a significant impact on the outcome of the algorithm. We use the transmit power budget defined in the 3GPP specification \cite{3GPPTS36213} as the reference to compute the initial PSD $\vp'$, such that the optimized solution of $(\vw, \vp)$ is guaranteed to provide a better performance than the standard configuration. The power spectral density in dBm (per RB) of link $l\in\overline{\Ks}$ is defined by $\text{PSD}_l = \min\{\text{PSD}_{\max}, \snr^{\text{tar}}_{l} + \text{P}_{\text{noise}}+\alpha\text{PL}_l\}$, where $\text{PSD}_{\max}$ denotes the maximum PSD, $\snr^{\text{tar}}_{l}$ is the open loop SNR target for the $l$th link,  $\text{P}_{\text{noise}}$ is the noise PSD in the receiver, $\alpha$ is the pathloss compensate factor, and $\text{PL}_l: = \text{PL}_{b_l, l}$ is the pathloss estimate of the link $l$ served by BS $b_l$.  
%
\label{rem:initialSelection}
\end{Rem}


\begin{algorithm}[t]
\SetKwData{band}{$\vw'$}
\SetKwData{pow}{$\vp'$}
\SetKwData{policy}{$\pi'$}
\SetKwFunction{UpdateBandwidth}{UpdateBandwidth}
\SetKwFunction{ScalePower}{ScalePower}
\SetKwFunction{UpdatePower}{UpdatePower}
\SetKwFunction{MinimizePower}{MinimizePower}
\SetKwInOut{Input}{input}
\SetKwInOut{Output}{output}
\caption{Joint Allocation of Bandwidth and Power}
\Input{$\vp'\leftarrow \hat{\vp}\in\RP^{2K}$, $\vw'\leftarrow \hat{\vw}\in\RP^{2K}$, $\vw\leftarrow\ve{0}$, $\lambda\leftarrow 0$, $\pi'\in\Pi$, $\epsilon_1$, $\epsilon_2$, $\epsilon_3$}
\Output{$\vw^{\star}$, $\vp^{\star}$}
\BlankLine
Compute $\mA^{\ul}(\pi')$, $\mA^{\dl}(\pi')$, $\mVt(\pi')$ and $\mD(\pi')$\;
 \emph{\% S1: Update $\vw$ based on Prop.\ref{prop:solu_FP_load}}\;
     \While{$\|\vw'-\vw\|_{\infty}\geq\epsilon_2$}{
		 $\vw\leftarrow\vw'$\;
		\emph{\% Fixed point iteration \eqref{eqn:FP_converge}}\;
      \band $\leftarrow$ \UpdateBandwidth{$\vp', \vw$}\;
     }
		 \emph{\% S2: Update $\vw$ to achieve full load  based on Prop.\ref{prop:fullload}}\;
 \If{$g_{1}(\vw')<1  \&  g_{2, \vp'}(\vw')=1$}{
\While{$g_{1}(\vw')<1$}{
$\vp\leftarrow\vp'$\;
\emph{\% Power scaling in \eqref{eqn:FP_scalePow}}\;
 \pow $\leftarrow$ \ScalePower{$\vw', \vp$}\;
     \While{$\|\vw'-\vw\|_{\infty}\geq\epsilon_2$}{
		 $\vw\leftarrow\vw'$\;
		\emph{\% Fixed point iteration \eqref{eqn:FP_converge}}\;
      \band $\leftarrow$ \UpdateBandwidth{$\vp', \vw$}\;
     }
 }
}
 \emph{\% S3: Update $\vp$}\;
 \If{$g_{1}(\vw')=1  \&  g_{2,\vp'}(\vw')<1$}{
$\vp\leftarrow\ve{0}$\;
		 \While{$\|\vp'-\vp\|_{\infty}\geq\epsilon_3$}{
$\vp\leftarrow\vp'$\;
\emph{\% Fixed point iteration \eqref{eqn:FP_converge_fixp}}\;
      \pow $\leftarrow$ \UpdatePower{$\vw',\vp$}\;
}
}
$\vw(\pi')\leftarrow \vw'$; $\vp(\pi')\leftarrow \vp'$; $\lambda(\pi')\leftarrow\lambda'$\; 
\label{algo:JointOptimization}
\end{algorithm}

\section{Numerical Results}\label{sec:simu}
In this section, we verify the propositions presented in Section \ref{sec:ResourceAllo} and \ref{sec:PowerControl},  show the convergence of Algorithm \ref{algo:JointOptimization},  and compare the performance with the proposed algorithm to the conventional resource allocation schemes under different association policies presented in Section \ref{subsec:LinkAssociation} through simulations.

\subsection{Simulation Parameters}
%
To obtain practically relevant results, we study the real-world scenario as shown in Fig. \ref{fig:mapBerlin}. This map shows the center of Berlin, Germany in the WGS
84 coordinate system. 
There are $81$ BSs, among which 45  of them are macro cell BSs ($1$ BS per sector) with directional antenna and maximum transmit power of $43$ dBm, while  36 of them are pico cell BSs with omni-directional antenna and maximum transmit power of $30$ dBm. 
We assume that a total bandwidth of $5$ MHz is subdivided into $25$ RBs of $12$ subcarriers each, and that the frequency reuse factor is $1$.
The color map refers to the pathloss in dB. For each pixel of $50 \times 50$m
size, the channel gain over all received downlink signals from the macro cell BSs is given according to the measured data of pathloss from \cite{MOMENTUM}. The pico cell BSs are randomly placed on the cell edge of the macro cells. Based on the 3GPP LTE model provided in \cite{3GPP36814}, we obtain the pathloss between the pico BSs and the UEs to compute $\mH_0$ (joint with the macro-to-UE pathloss), 
the pathloss between the BSs to compute $\mH_1$, and the pathloss between the mobile terminals to compute $\mH_2$. On top of this realistic pathloss, we implement uncorrelated fast fading characterized by Rayleigh distribution.
We assume reciprocal uplink and downlink channels. 

The users are uniformly randomly distributed in the playground. The maximum transmit power of the user terminal is $22$ dBm. We define $5$ service classes, with the downlink rate requirements of $[300, 25, 50, 10, 0.01]$ Mbit/s, and the corresponding uplink rate requirements of $[50, 50, 25, 10, 0.01]$ Mbit/s. These classes imply the following $5$ services: 1)  cloud service video and other digital service, 2) HD video/photo sharing, 3) high-resolution video and other digital services, 4) broadband data allowing video email and web surfing, and 5) text, voice or video messages.


\begin{figure}[t]
\centering
\includegraphics[width=1\columnwidth]{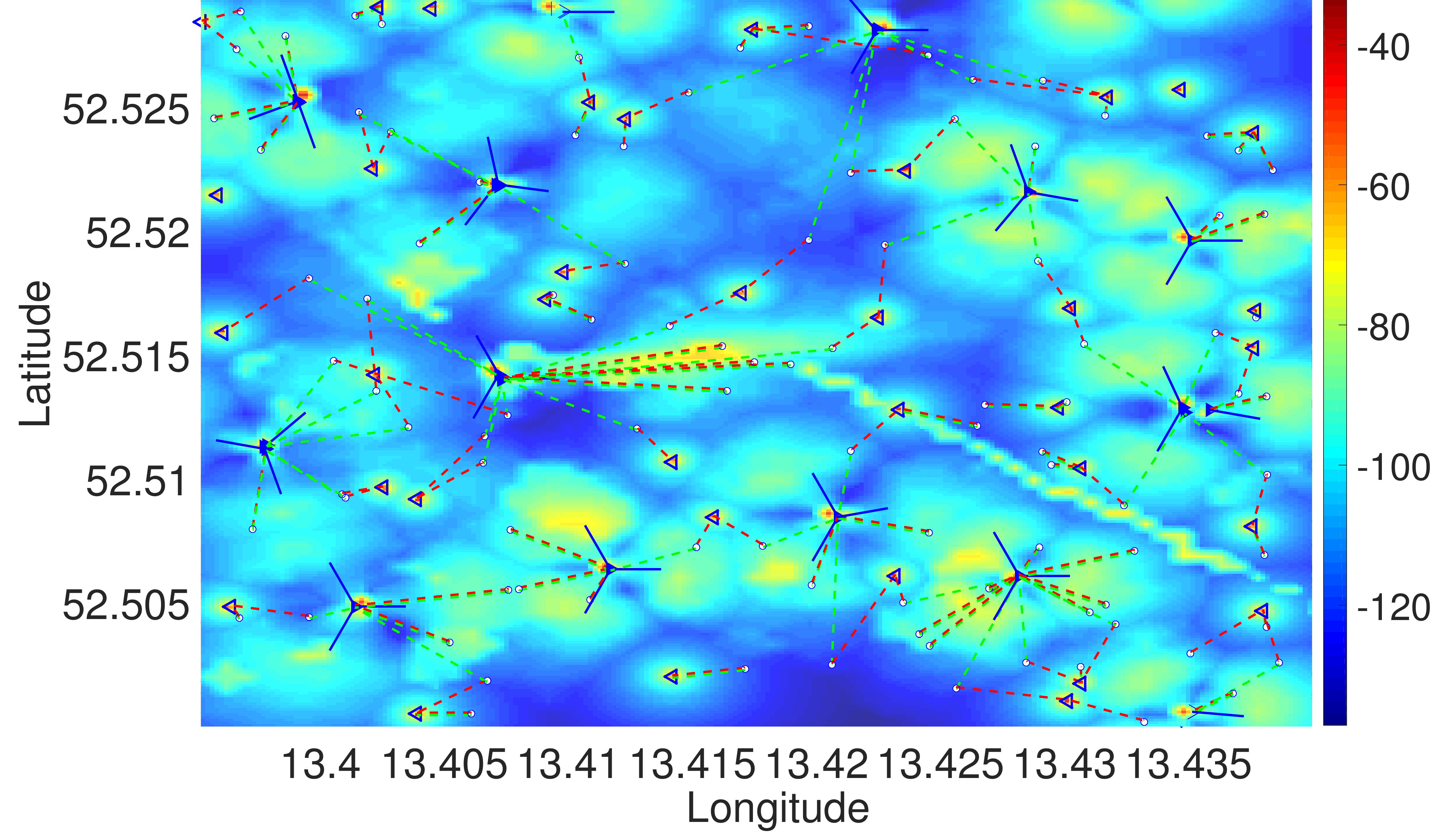}
\caption{DeUD-enabled wireless network. Macro BSs - blue solid triangles; pico cells - blue hollow triangles; UEs - white circle with blue edge; downlink association - green dashed line; uplink association - red dashed line.}
\label{fig:mapBerlin}
\end{figure}

\subsection{Convergence of The Algorithm}
Let us first examine the convergence behavior of the algorithms presented in Prop. \ref{prop:solu_FP_load}, \ref{prop:fullload} and \ref{eqn:prop_power_allo} (corresponding to S1, S2, and S3) in Algorithm \ref{algo:JointOptimization}, respectively. 
In Fig.\ref{fig:AlgorConverge} we verify the propositions and show the convergence of the algorithm \ref{algo:JointOptimization} with the fixed association policy DeUD\_P, at a single simulation snapshot (i.e., the users are assumed to be static within one time interval). 
The number of users is $K = 500$.
The desired numerical precisions are set to $\epsilon_i = 1e-7$, for $i = 1,2,3$.

Fig. \ref{fig:FP_updatingP} illustrates the convergence behavior of three successive steps S1, S2, and S3. The algorithm starts at step S1, where $g_1(\vw^{(0)})<1$ and $g_2(\vp^{(0)}, \vw^{(0)})<1$. The initial power $\vp^{(0)}$ is chosen as described in Rem. \ref{rem:initialSelection}, where $\text{PSD}_{\max} = 12$ dBm, $\snr^{\text{tar}} = 12.2$ dB, $\alpha = 1$, and $P_{\text{noise}} = -121.45$ dBm. The initial bandwidth allocation  is defined as $\vw^{(0)} = \ve{0}$.
After performing the fixed point iteration \eqref{eqn:FP_converge} at S1, it converges to the fixed point $\vw'$ such that $g_2(\vp^{(0)}, \vw') = 1$ while $g_1(\vw')$ is extremely small (approximately $0.01$). The algorithm moves therefore to S2 of power scaling. The algorithm at S2 converges to the point $(\vw'', \vp')$, where $g_1(\vw') =1$ and $g_2(\vw'', \vp') <1$, which causes the algorithm to move to S3. By the end of S3, the fixed point iteration \eqref{eqn:FP_converge_fixp} converges to $\vp''$ such that $g_1(\vw'') = g_2(\vw'',\vp'') = 1$, and the algorithm terminates. At each step, the iteration improves the desired utility $\lambda$ monotonically.

An interesting observation we have made concerning the relationship between per-cell power constraint and the feasible utility is illustrated in Fig. \ref{fig:PowerUpdating}. The motivation is to find out the tradeoff between the power consumption and the improvement of the utility. Fig. \ref{fig:PowerUpdating} shows the increase of the utility as we increase the power constraint factor $\theta$ ($\theta$ increases from $0.01$ to $1.01$ with step size of $0.01$), under different self-noise power $\sigma$. As shown in Thm. \ref{the:MSS_FP}, $\theta$ is the scaling factor of the monotonic constraint $g(\vx)$. As for S3, in particular, $\theta$ is scaling factor of the maximum power constraint such that $g_{2,\vw'}(\vp)\leq\theta$. For small value of $\sigma$ (i.e., in an interference-dominant system), small value of $\theta$ is sufficient for the feasible utility, and increase of $\theta$ only leads to minor increase of utility (blue and red curves for the noise power of $-121$ dBm and $-100$ dBm, respectively). Conversely, for the large value of $\sigma$ (i.e., in a noise-dominant system),  increase of $\theta$ has a stronger effect on improving utility (green and black curves for the noise power of $-80$ dBm and $-70$ dBm, respectively). The above observation can help us to choose a proper operation point, to provide a good tradeoff between the total power consumption and the desired utility.

Fig. \ref{fig:UEvsCellSpecific} and \ref{fig:EnergyEff} are provided to illustrate the performance of algorithms presented in Section \ref{subsec:percell_powercontrol} and \ref{subsec:efficient_powercontrol}. Fig. \ref{fig:UEvsCellSpecific} shows a case that restricting cell-specific DL power results in approximately $16\%$ degradation of utility achieved by UE-specific DL power. Fig. \ref{fig:EnergyEff} shows a specific example that for a certain snap shot of the network, over $90\%$ of power consumption can be saved if we only target at required utility $\lambda = 1$ instead of the maximum feasible $\lambda$, by performing the step of energy efficient power control presented in Section \ref{subsec:efficient_powercontrol}.

\begin{figure}[t]
\centering
\includegraphics[width=.8\columnwidth]{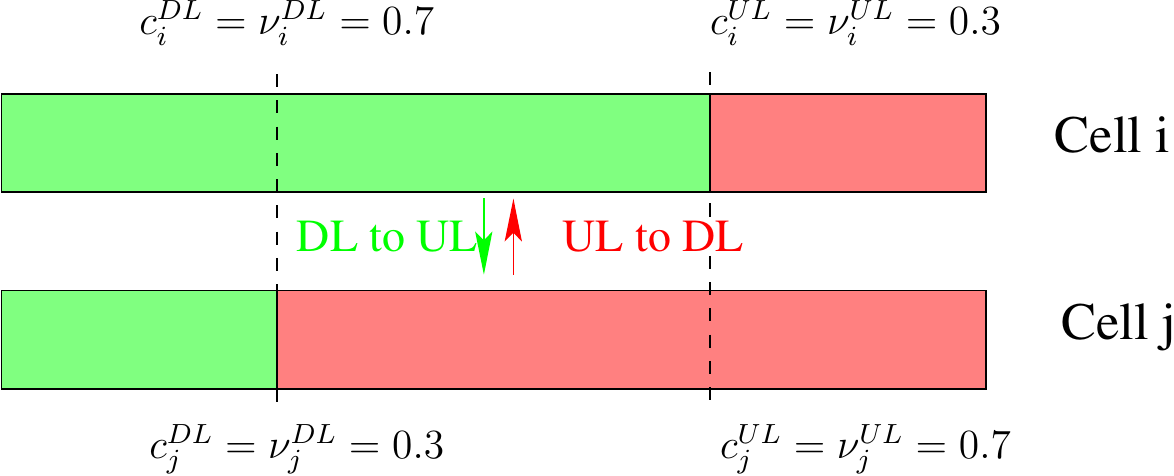}
\caption{One possible approach to estimate the overlap factor based on the historical load measurements. The overlap factor between downlinks served by cell $i$ and the uplinks served by cell $j$ is computed by $c_i^{\dl}c_j^{\ul} = 0.49$, while the overlap factor between the uplinks served by cell $i$ and the downlinks served by cell $j$ is computed by $c_i^{\ul}c_j^{\dl} = 0.09$.}
\label{fig:overlappingfactor}
\end{figure}

%
%
%

\subsection{Network Performance Evaluation}
\subsubsection{Selection of Association Policy}\label{subsubsec:SelectAssociation}
Now let us examine the performance of Algorithm  \ref{algo:JointOptimization} under different link association policies. 
The set of association policies $\Pi$, including CoUD, DeUD\_O (with variety of offsets) and DeUD\_P as introduced in Section \ref{subsec:LinkAssociation}, is defined as follows. Note that all macro cell BSs have maximum transmit power  of 43 dBm, while all small cell BSs of $30$ dBm. Thus, by setting $\text{offset}^{\text{UL}}_n=13$ dB for $n$ as small cell BS, the policy DeUD\_O is equivalent to DeUD\_P, while by setting $\text{offset}^{\text{UL}}_n = 0$ for all $n\in\Ns$, the policy DeUD\_O is equivalent to CoUD. The set of policy $\Pi$ is then defined as a set of DeUD\_O policies with offsets $\{0, 1,3, 5, \ldots, 51\}$ of the small cell BSs in UL, where $0$ corresponding to CoUD and $13$ corresponding to DeUD\_P.

Fig. \ref{fig:DistributionOffset} shows the average performance of the algorithm under each policy $\pi\in\Pi$ using the Monte Carlo techniques. We run 500 independent tests,  with uniform user distribution of 100 static users in each test. Fig. \ref{fig:probselectOff} shows the percentage of the counts that a fixed policy provides the utility among the top three maximum utilities achieved by all policies. Fig. \ref{fig:UtilityCI_UE100} shows the average utility of a fixed policy over the $500$ tests (the high value of utility is due to the lower number of the users compared to Fig. \ref{fig:AlgorConverge}). The following two observations are made. 1) Proper selection of DeUD policy can achieve approximately $2\times$ improvements on desired utility, compared against CoUD. 2) Although DeUD\_P is not always the best policy that provides maximum utility, it has a high chance to provide relatively good performance (approximately $73\%$ of counts among the top three maximum utility).
Thus, in case the operator wants to save the computational cost of exhaustive searching for optimal association policies, always selecting DeUD\_P provides a suboptimal compromises. However, we shall remind that in many cases, DeUD\_P is not the best association policy with respect to maximizing the desired utility, as shown in the two examples of the single trial in Fig. \ref{fig:UECurve_1} and Fig. \ref{fig:UECurve_2} respectively.

\subsubsection{Effects of Overlapping Uplink/Downlink Frequency Bands}
Note that in Section \ref{subsubsec:SelectAssociation}, the frequency band allocation follows the rule that only {\it partial overlap} between UL/DL frequency band is allowed to mitigate the inter-link interference, as shown in Rem. \ref{Rem:pverlapfactor}. Computation of the overlap factor is provided by Appendix \ref{subsec:Overlap}. 
Since the overlap factor is estimated based on the historical measurements, the actual utility $\lambda$ derived using optimized $(\vp, \vw)$ may not be as high as the computed $\lambda$ in Algorithm \ref{algo:JointOptimization}. 
On the other hand, if {\it full overlap} is allowed (i.e., each transmission can be allocated to any of the RBs, regardless of whether it is in UL or DL), then, the overlap factor is one, and the utility achieved by Algorithm \ref{algo:JointOptimization} can be much lower due to the strong inter-link interference.

In Fig. \ref{fig:JoULDL} we show the utility achieved by our proposed joint UL/DL optimization algorithm (represented by \lq\lq Jo\rq\rq), with the strategy of partial or full overlap. The three subplots from left to right illustrate the utility when the association policies \lq\lq Best\rq\rq, \lq\lq DeUD\_P\rq\rq and  \lq\lq CoUD\rq\rq are applied, respectively. Policy \lq\lq Best\rq\rq \ denotes the policy where the offset provides the maximum value of $\lambda$, i.e., $\pi^{\star} = \argmax_{\pi\in\Pi}\lambda(\pi)$. 
For scenario of partial overlap, the blue dashed line expresses the optimized $\lambda$ computed with our algorithm, while the green and red solid lines express the actual $\lambda$ in UL and DL, respectively. Although the algorithm aims at achieving fair user-specific UL and DL utility, a small gap between the UL and DL utility can be observed due to the biased estimation of the overlap factor. 
For scenario of full overlap, the magenta solid line expresses the achieved $\lambda$ for both UL and DL. Because the interference coupling model in \eqref{eqn:virtualSINR} is accurate under the assumption of full overlap, there is no gap between the computed $\lambda$ and the actual achievable $\lambda$.

Furthermore, we make the following observations. 1) Using optimized $(\vw, \vp)$ based on estimated overlap factor, we can achieve the actual utility in DL only about $2\%-3\%$ lower than the computed maximum feasible $\lambda$ from the proposed algorithm, and in UL about  $10\%-30\%$ lower. 2) By regulating the frequency band allocated to UL and DL transmission with partial overlap, we achieve a $50\%-100\%$  increase in utility than allowing the full overlap. 3) By enabling UL and DL decoupling, we can achieve a two-fold increase in the utility, compared to CoUD. Although  DeUD\_P may not be the best association policies, it still provides $60\%-75\%$ increase. The same conclusion is reached by the analysis on association policies in Section \ref{subsubsec:SelectAssociation}.
\subsubsection{Comparison against QoS-Based Proportional Fairness}
We use the proportional fairness (PF) algorithm as a baseline for evaluating the utility benefits provided by our algorithm. To provide a fair comparison between the PF algorithm and our proposed algorithm, instead of the rate-based PF algorithm \cite{nguyen2006proportional}, we replace the rate with the metric of level of QoS satisfaction, i.e., $W_0w_lr_l/d_l$ for link $l\in\overline{\Ks}$ presented in \eqref{eqn:utility_lambda}. 
We run PF algorithm under default UL/DL bandwidth ratio under both association policies CoUD and DeUD\_P, to compare with the proposed joint UL/DL optimization algorithm. The default UL/DL bandwidth ratio is set to be $9:16$, i.e., out of $25$ RBs, $9$ of them are assigned for UL transmission while $16$ for DL transmission.

Fig. \ref{fig:comparisonPF} shows the performance comparison between our proposed algorithm and the PF algorithm under DeUD\_P and CoUD.  
Conventional PF algorithm achieves fairness in UL and DL independently, and the fixed ratio of UL/DL bandwidth ratio causes a large gap between the achievable utility in UL and DL. Our proposed Algorithm \ref{algo:JointOptimization} outperforms the PF algorithm, in the sense that it jointly optimizes the level of QoS satisfaction in UL and DL to the  best closing levels. The utility in UL achieves three-fold increase than the PF algorithm in both DeUD\_P and CoUD. We still observe a $20\%-50\%$ increase in DL utility in DeUD\_P, while in CoUD we sacrifice some DL utility to achieve a higher gain in UL. However, as more UEs are served in the system, even in CoUD we achieve better utility in both UL and DL than the QoS-based PF algorithm.  

Another observation in reference to Fig. \ref{fig:comparisonPF} is that, for both algorithms, by splitting the UL/DL access, the performance can be further improved by about $60\%-70\%$. It is worth mentioning that the gain of UL/DL decoupling is not as high as expected in \cite{boccardi2015decouple,elshaer2014downlink} (more than two-fold increase). Our explanation is that although the strength of the useful signal is increased by offloading more uplinks in small cells, the received signal strength of the interference may also be increased because the small cells are normally located on the cell edge. Therefore, it increases the need for the joint UL/DL optimization algorithm allowing flexible UL/DL bandwidth ratio, as we proposed in Algorithm \ref{algo:JointOptimization}.

%


\section{Conclusion}\label{sec:concl}
We studied the utility maximization problem for the uplink and downlink decoupling-enabled HetNet, to jointly optimize the uplink and downlink bandwidth allocation and power control, under different association policies. The utility is modeled as the minimum level of the QoS satisfaction, to achieve fair service-centric performance. We develop a general model of inter-cell interference, that includes inter-link interference between uplink and downlink, with properties of power coupling and load coupling. Based on the interference model, we develop a three-step optimization algorithm  using the fixed point approach for nonlinear operators with or without monotonicity. The algorithm benefits from the user-centric context-aware communication environment in 5G networks, adapts the bandwidth allocation and power spectral density according to the channel condition and traffic demand in both UL and DL, and achieves jointly optimized utility in both UL and DL. Numerical results show that the performance of our algorithm outperforms the QoS-based proportional fairness algorithm, and it is robust against heavily loaded system with high traffic demand.

\appendix
\label{appendix}

\subsection{Approximation of Overlap Factor}\label{subsec:Overlap}

One possible method is to compute the overlap factor proportional to the fraction of the overlapping band. For example, the cell-pairwise directional overlap factor $o_{i,j}^{\text{X}\leftarrow \text{Y}}$  for $\text{X}, \text{Y}\in\{\ul, \dl\}$ and $i,j\in\Ns, i\neq j$ can be define by $o_{i,j}^{\text{X}\leftarrow \text{Y}}:= \max\{0, (\nu_j^{\text{Y}}+\nu_i^{\text{X}}-1)/\nu_i^{\text{X}}\}$ if $\text{X}\neq\text{Y}$, to express the probability that a RB in cell $i$ receives interference in UL (DL) from any  DL (UL) transmission  signal in cell $j$ (inter-cell inter-link interference); and $o_{i,j}^{\text{X}\leftarrow \text{Y}}:=\max\{1, \nu_j^{\text{Y}}/\nu_i^{\text{X}}\}$ if $\text{X}=\text{Y}$, to express the probability that a RB in cell $i$ receives interference in UL (DL) from any  UL (DL) transmission signal in cell $j$ (inter-cell intra-link interference). For example, assuming $\nu_i^{\dl} = 0.7, \nu_i^{\ul} = 0.3$ for cell $i$ and $\nu_j^{\dl}=0.3, \nu_j^{\ul} = 0.7$ (as shown in Fig. \ref{fig:overlappingfactor}), we have $o_{i,j}^{\dl\leftarrow\ul} = \max\{0, (\nu_j^{\ul}+\nu_i^{\dl}-1)/\nu_i^{\dl}\} = \max\{(0.7+0.7-1)/0.7, 0\} \approx 0.57$, while $o_{ij}^{\ul\leftarrow\dl} = \max\{0, (\nu_j^{\dl}+\nu_i^{\ul}-1)/\nu_i^{\ul}\} = 0$. Let us define the overlap matrix $\ma{O}^{\text{X}\leftarrow \text{Y}}:=(o_{i,j})^{\text{X}\leftarrow \text{Y}}\in[0,1]^{N\times N}$,  for $\text{X}, \text{Y}\in\{\ul, \dl\}$. To transform  $\ma{O}^{\text{X}\leftarrow \text{Y}}$ to the per-link basis matrix (between the UL and DL), we define   $\tilde{\ma{O}}^{\text{X}\leftarrow \text{Y}}: = (\mA^{\text{X}})^T\ma{O}^{\text{X}\leftarrow \text{Y}}\mA^{\text{Y}}$.
The cross-link coupling matrix is then modified by computing the Hadamard product (element-wise product) of $\mVt^{\text{X}\leftarrow \text{Y}}$ and $\tilde{\ma{O}}^{\text{X}\leftarrow \text{Y}}$, for $\text{X}, \text{Y}\in\{\ul, \dl\}$.


Unfortunately, the fraction of the overlapping bands depends on the cell-specific loads $\ve{\nu}^{\ul}$ and $\ve{\nu}^{\dl}$, which further depend on the  dynamic UL and DL resource allocation $\vw$ (as the variable to be optimized in Prob. \ref{prob:jointOpt}). Thus, introducing such a modification dramatically complicates optimization problem. 

A compromise approach is to use the historical measurements of load $\ve{\nu}^{\ul}$ and $\ve{\nu}^{\dl}$ as estimates to compute the {\it cell-pairwise overlap factor} $o_{ij}^{\text{X}\leftarrow \text{Y}}$ for $\text{X}, \text{Y}\in\{\ul, \dl\}$, $i,j\in\Ns$ as described above. 

An alternative to the {\it cell-pairwise overlap factor} $o_{ij}^{\text{X}\leftarrow \text{Y}}$ is to define a {\it cell-specific overlap factor} $c_i^{\text{X}}$, for $\text{X}\in\{\ul, \dl\}$, $i\in\Ns$ to express how likely a transmission in cell $i$ causes inter-link interference to the transmission in another cell, while the computation of intra-link overlap factor remains the same as the approach above. This approach is more error-tolerant in the sense that it does not return zero probability for inter-cell inter-link interference. 
%
%
We define two vectors with constant values $\vc^{\ul}\in [0,1]^N$ and $\vc^{\dl}\in [0,1]^N$, which can be chosen proportional to the historical measurements of $\ve{\nu}^{\ul}$ and $\ve{\nu}^{\dl}$, respectively.
Further we can modify the cross-link coupling matrix by defining $\mV^{\dlul} \coloneqq  (\mA^{\ul})^T\diag(\vc^{\ul})\mH_1\diag(\vc^{\dl}){\mA^{\dl}}$, and $\mV^{\uldl}  \coloneqq  \diag\left((\mA^{\dl})^T\vc^{\dl}\right)\mH_2\diag\left((\mA^{\ul})^T\vc^{\ul}\right)$, such that the coupling between UL and DL is proportional to the multiplication of the cell UL and DL overlap factors. 
For example, the overlap factor between the downlinks in cell $i$ and the uplinks in cell $j$ is proportional to $c^{\dl}_i c^{\ul}_j$ as shown in Fig. \ref{fig:overlappingfactor}.

\subsection{Standard Interference Function}\label{subset:SIF}
\begin{definition}
A vector function $\vf:\RN^k\to \RP^k$ is said to be a standard interference function (SIF) if the following axioms hold:
\begin{itemize}
\item[1.] (Monotonicity) $\vx\leq \vy$ implies $\vf(\vx)>0\leq\vf(\vy)$
\item[2.] (Scalability) for each $\alpha>1$, $\alpha\vf(\vx)>\vf(\alpha\vx)$ 
\end{itemize}
\label{def:SIF}
\end{definition}
The original definition of standard interference function is stated in \cite{Yates95a}, which also requires positivity. In Definition \ref{def:SIF} we drop the positivity $\vf(\vx)>0$ for $\vx\in\RN^k$ because it is a consequence of the other two properties \cite{leung2004convergence}. 
\begin{lemma}[Selected Properties of SIF \cite{Yates95a}]
Let $\vf:\RN^k\to\RP^k$ be a SIF. Then 
\begin{itemize}
\item[1.] There is at most one fixed point $\vx\in \fix(\vf) \coloneqq \{\vx\in\RP^k|\vx = \vf(\vx)\}$.
\item[2.] The fixed point exists if and only if there exists $\vx'\in\RP^k$ satisfying $\vf(\vx')\leq\vx'$.
\item[3.] If a fixed point exists, then it is the limit of the sequence $\{\vx^{(n)}\}$ generated by $\vx^{(n+1)} = \vf(\vx^{(n)})$, $n\in\NN$, where $\vx^{(1)}\in\RN^k$ is arbitrary. If $\vx^{(1)} = \ve{0}$, then the sequence is monotonically increasing (in each component). In contrast, if $\vx^{(1)}$ satisfies $\vf(\vx^{(1)})\leq\vx^{(1)}$, then the sequence is monotonically decreasing (in each component). 
\end{itemize}
\label{lem:Prop_SIF}
\end{lemma}

\subsection{Proof of Lemma \ref{lem:StandardInterference}}\label{subsec:proofSIF}
The essential steps of the proof follow those in the proof of \cite[Ex. 2]{Cavalcante14}. First we show that $f_{\vp', l}(\vw) \coloneqq d_l/\left(W_0 B\log(1+\sinr_l(\vw))\right)$ is positive and concave. 
Function $f_{\vp', l}(\vw)$ is positive concave, because of the following facts: \rmnum{1}) $h(x) \coloneqq  1/\log_2(1+1/x)$ is a concave function on $\RP$,  \rmnum{2}) composition of concave functions with affine transformations (see the interference term in \eqref{eqn:virtualSINR}) preserves concavity, and \rmnum{3}) a set of concave functions is closed under multiplication and addition. 
Then, because a positive concave function is proved to be a SIF in \cite[Prop. 1]{Cavalcante14}, $f_{\vp', l}$ is SIF. As a collection of $\{f_{\vp', l}\}$, the vector function $\vf_{\vp'}$ is SIF.
 %
\subsection{Proof of Theorem \ref{the:MSS_FP}}\label{subsec:proofThe_1}
Since the essential steps follow those in the proof of \cite[Th. 3.2]{nuzman2007contraction}, we describe only proof outlines and mention crucial lemmas in this paper, for lack of space. Using \cite[Lem. 3.3]{nuzman2007contraction}, we know that $\vh \coloneqq \vx/g(\vx)$ is non-expansive on $(\RP^k, \mu_s)$, where the metric $\mu_s$ is defined as $\mu_s(\vx,\vy) \coloneqq \max\limits_{i=1, \ldots, k}(\log(x_i/y_i))^{+}+\max\limits_{i=1, \ldots, k}(\log(y_i/x_i))^{+}$. Because $\vf$ is SIF, by virtue of \cite[Lem. 2.2]{nuzman2007contraction}, $\vps = \theta \vh\circ\vf = \theta\vf/(g\circ\vf)$ in \eqref{eqn:FP_converge} is shrinking (or contractive) with respect to $\mu_s$.

If $\vps$ is a contractive mapping on a compact metric space on $(\RP^k, \mu_s)$, there exists a unique fixed point $\vx\in\RP^k$ with $\vps(\vx)=\vx$ \cite[Th.5.2.3]{smart1980fixed}. In the following we show that $\vps$ is a mapping of a compact space to itself. For any input, since $g$ is homogeneous on $\RP^k$, we have $g\circ\vps = (\theta/g\circ\vf) \cdot (g\circ\vf) = \theta$. Because a monotonic vector function has bounded level sets, we have that $\vps(\vx)\leq \vb$ for some finite $\vb>\ve{0}$. With $\vps(\vx)\leq \vb$ and $\vf(\vx)\geq \vf(\ve{0})$ for all $\vx\in\RN^{k}$, we have $\vps^2(\vx) \geq\theta\vf(0)/(g\circ \vf(\vb))= \va > \ve{0}$, and we see that the range of $\vps^n$ falls inside the finite positive rectangle $\Rec(\va,\vb)$ for $n\geq 2$. Hence, there is exactly one eigenvector $\vx\in\RP^k$ to satisfy $\vx' = \rho' \vf(\vx')$ where the associate eigenvalue is given by $\rho' = \theta/(g\circ\vf(\vx'))$, such that $g(\vx') = g(\vps(\vx')) = \theta$. 

\subsection{Proof of Prop. \ref{prop:fullload}}\label{subsec:proofProp_2}
%
We will prove by induction that by using algorithm in Prop. \ref{prop:fullload}, the sequence $\lambda$ is monotonically increasing until $g_1(\vw) = 1$ is satisfied.

At the base step, suppose the solution to P.2a yields $\vw' = \lambda'\vf_{\vp'}(\vw')$ where $\lambda' \coloneqq  1/g_{\vp'}(\vw')$ and $g_{\vp'}(\vw') = \max\{g_1(\vw'), g_{2,\vp'}(\vw')\}$, with $g_1(\vw')<1$ and $g_{2, \vp'}(\vw') = 1$. 
Let us define $g_1(\vw')= a<1$ and $\vp''=a\vp'$. With fixed $\vp''$, using Theorem \ref{the:MSS_FP}, iteration \eqref{eqn:FP_converge} converges to a unique fixed point $\vw''$, satisfying 
\begin{align}
\vw'' & = \lambda''\vf_{\vp''}(\vw'')\label{eqn:prop_2_3}\\
\mbox{such that } & \max\{g_1(\vw''), g_2(\vp'', \vw'')\}=1\label{eqn:prop_2_4}
\end{align}


It is clear that $\vf_{\vp''}(\vw')<\vf_{\vp'}(\vw')=\vw'/\lambda'$, by dividing both the numerator and denominator by $a$ in \eqref{eqn:virtualSINR}, and substituting \eqref{eqn:virtualSINR} in \eqref{eqn:rate} and \eqref{eqn:prob_1_def_load}.  Now let us define $\vv' = \vw'/a >\vw'$. Moreover, knowing that $\vf_{\vp''}$ is also a SIF, we have $\vf_{\vp''}(\vv')=\vf_{\vp''}(\vw'/a)<\vf_{\vp''}(\vw')/a$ due to the scalability, that further leads to $\vf_{\vp''}(\vv')<\vf_{\vp''}(\vw')/a<\vf_{\vp'}(\vw')/a=\vw'/(a\lambda') = \vv'/\lambda'$. In other words, there exists $\vv'$ such that $\lambda'\vf_{\vp''}(\vv')<\vv'$, and $\vv'$ is a feasible point with respect to the SIF $\vf'_{\vp''}  \coloneqq  \lambda'\vf_{\vp''}$. Thus, starting from $\vv'$, the sequence of $\vv$ decrease monotonically to a unique fixed point (by using the third property of SIF stated in Lemma \ref{lem:Prop_SIF})
\begin{equation}
\vv''= \vf'_{\vp''}(\vv'')<\vf'_{\vp''}(\vv')<\vv'
\label{eqn:Prop_2_1}
\end{equation}
Due to the monotonicity and homogeneity of $g_1$ with respect to $\vw$, and the same properties of $g_2$ with respect to both $\vp$ and $\vw$, we have
\begin{align}
g_1(\vv'')& <g_1(\vv')=g_1(\vw'/a) =g_1(\vw') a=1\label{eqn:Prop_2_2a}\\
g_2(\vp'', \vv'')&<g_2(a\vp', \vv') = g_2(a\vp',\vw'/a)=1
\label{eqn:Prop_2_2b}
\end{align}

We prove $\lambda''>\lambda'$ by contradiction. Suppose $\lambda''\leq\lambda'$, then we have $\lambda''\vf_{\vp''}(\vv'')\leq\lambda'\vf_{\vp''}(\vv'')=\vv''$, using \eqref{eqn:Prop_2_1}. By defining $\vf''_{\vp''} \coloneqq \lambda''\vf_{\vp''}$ which is also a SIF, since $\vf''_{\vp''}(\vv'')\leq\vv''$, starting from $\vv''$, the sequence of $\vw$ is monotonically decreasing to the unique fixed point $\vv^\star$ satisfying $\vv^\star = \vf''_{\vp''}(\vv^\star) = \lambda''\vf_{\vp''}(\vv^\star)$. Because $\vv^{\star}$ is unique (by using the first and second properties of SIF stated in Lemma \ref{lem:Prop_SIF}), using \eqref{eqn:prop_2_3}, we have $\vw'' = \vv^\star\leq\vv''$, which further leads to $\max\{g_1(\vv''), g_2(\vp'', \vv'')\}\geq\max\{g_1(\vw''), g_2(\vp'', \vw'')\}=1$. This contradicts the inequalities \eqref{eqn:Prop_2_2a} and \eqref{eqn:Prop_2_2b}. Thus, we have that $\lambda''>\lambda'$ if $g_1(\vw')<1$.


For the further iteration step, using \eqref{eqn:prop_2_4}, it remains to consider cases in which $g_1(\vw'') = 1$, or $g_1(\vw'')<1, g_2(\vp'', \vw'')=1$. The former case directly leads to  $g_1(\vw'')=1$, and the algorithm stops at $\lambda''>\lambda'$. The latter case yields $g_1(\vw'')<1$, The proof above shows that the iteration step further increases $\lambda$, with scaled $\vp''' = g_1(\vw'')\vp''$.

\subsection{Proof of Prop. \ref{eqn:prop_power_allo}}\label{subsec:proofProp_3}
The solution to P.2a satisfies $\vp' = \lambda' \vf_{\vw'}(\vp')$ using the reformulation in \eqref{eqn:reformulate_w}. Since the variables $\vp$ and $\vw$ are interchangeable in $g_{2}$, we have $g_{2, \vp'}(\vw') = g_{2, \vw'}(\vp')$.

Therefore, if $g_{2, \vw'}(\vp')  = 1$, Theorem \ref{the:MSS_FP} implies that there is exactly one eigenvector $\lambda$ and associate eigenvector $\vp$ of $\vf_{\vw'}$ such that $g_{2, \vw'}(\vp') = 1$, and we have $\lambda'' = \lambda'$ and $\vp'' = \vp'$.

Then we consider the case when $g_{2, \vw'}(\vp')  < 1$. Because $\vp''$ is the optimal solution to P.2b, if we can find a $\hat{\vp}\in\RP^{2K}$ such that $\hat{\lambda} := \min_{l\in\overline{\Ks}} \hat{p}_l/f_{\vw', l}(\hat{\vp})$, $g_{2, \vw'}(\hat{\vp})\leq 1$ and $\hat{\lambda}>\lambda'$, then we have $\lambda''\geq\hat{\lambda}>\lambda'$. Thus, the remaining task is to find an arbitrary $\hat{\vp}$ that fulfills the above mentioned conditions. Let us define $\alpha = 1/g_{2, \vw'}(\vp')>1$ and $\hat{\vp} := a \vp'$. Then, we have 
\begin{equation*}
\hat{\lambda} = \min\limits_{l\in\overline{\Ks}} \frac{\alpha p'_l}{f_{\vw', l}(\alpha\vp')}>
\min\limits_{l\in\overline{\Ks}} \frac{\alpha p'_l}{\alpha f_{\vw', l}(\vp')} = \lambda'
\end{equation*}
The above inequality is due to the scalability of the SIF $\vf_{\vw'}$.

\section*{Acknowledgment}
We would like to thank Renato L.G. Cavalcante, Carl Nuzman and Paolo Baracca for the numerous technical discussions.

\begin{figure}[t]
    \centering
		    \begin{subfigure}[t]{1\columnwidth}
        \centering
        \includegraphics[width=1\columnwidth]{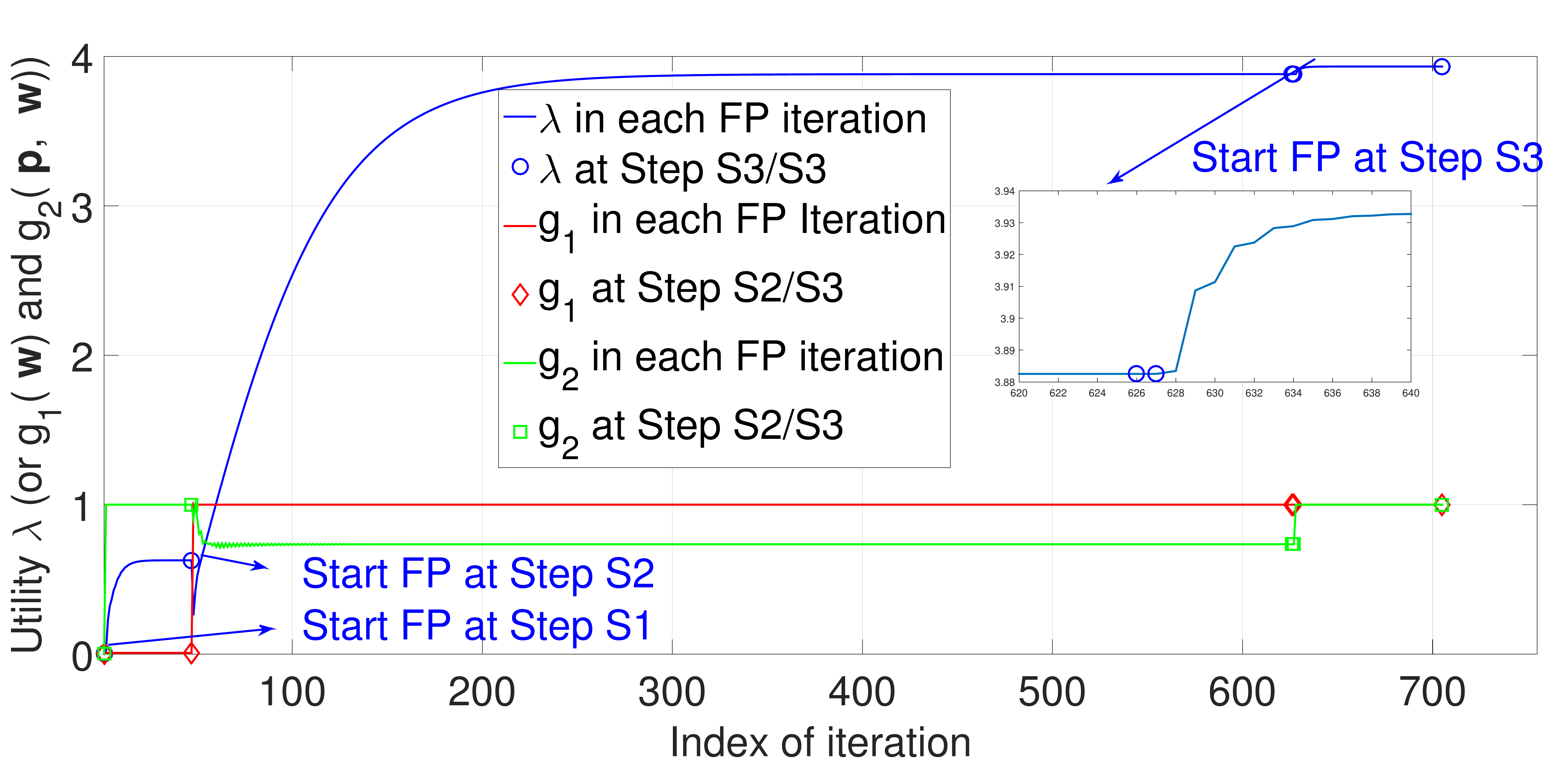}
        \caption{Convergence of Algorithm \ref{algo:JointOptimization}.}
 \label{fig:FP_updatingP}
    \end{subfigure}
    \begin{subfigure}[t]{1\columnwidth}
        \centering
        \includegraphics[width=1\columnwidth]{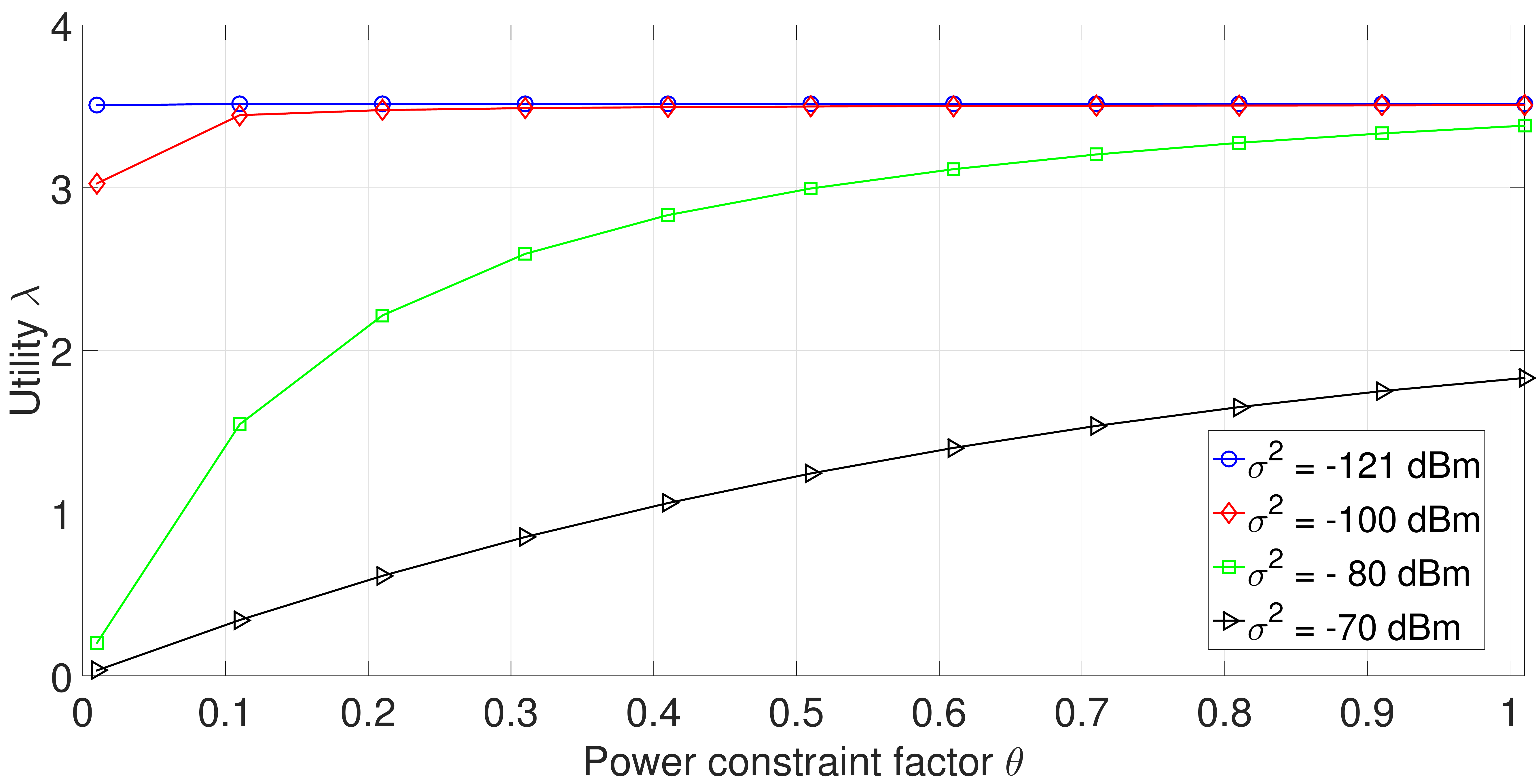}
        \caption{Dependence of optimized utility at S3 on $\theta$ and $\sigma^2$.}
 \label{fig:PowerUpdating}
    \end{subfigure}
    \begin{subfigure}[t]{1\columnwidth}
        \centering
        \includegraphics[width=1\columnwidth]{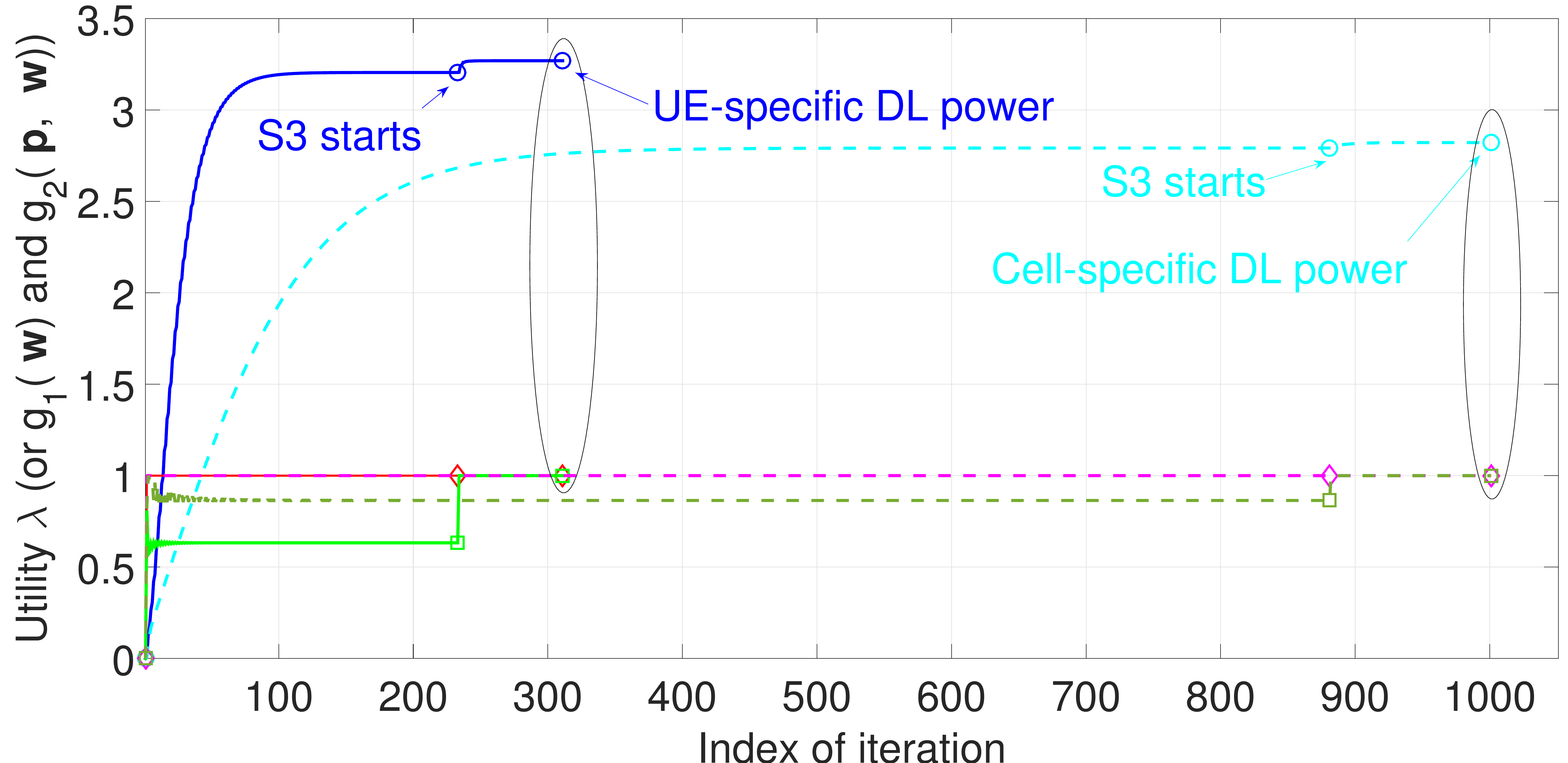}
        \caption{Comparison between UE-specific power control and cell-specific power control in DL. UE-specific: $\lambda$ - blue solid line, $g_1(\vw)$ - red solid line, $g_2(\vw, \vp)$ - green solid line. Cell-specific: $\lambda$ - cyan dashed line,  $g_1(\vw)$ - magenta dashed line, $g_2(\vw, \vp)$ - dark green dashed line.}
				\label{fig:UEvsCellSpecific}
    \end{subfigure}
		    \begin{subfigure}[t]{1\columnwidth}
        \centering
        \includegraphics[width=1\columnwidth]{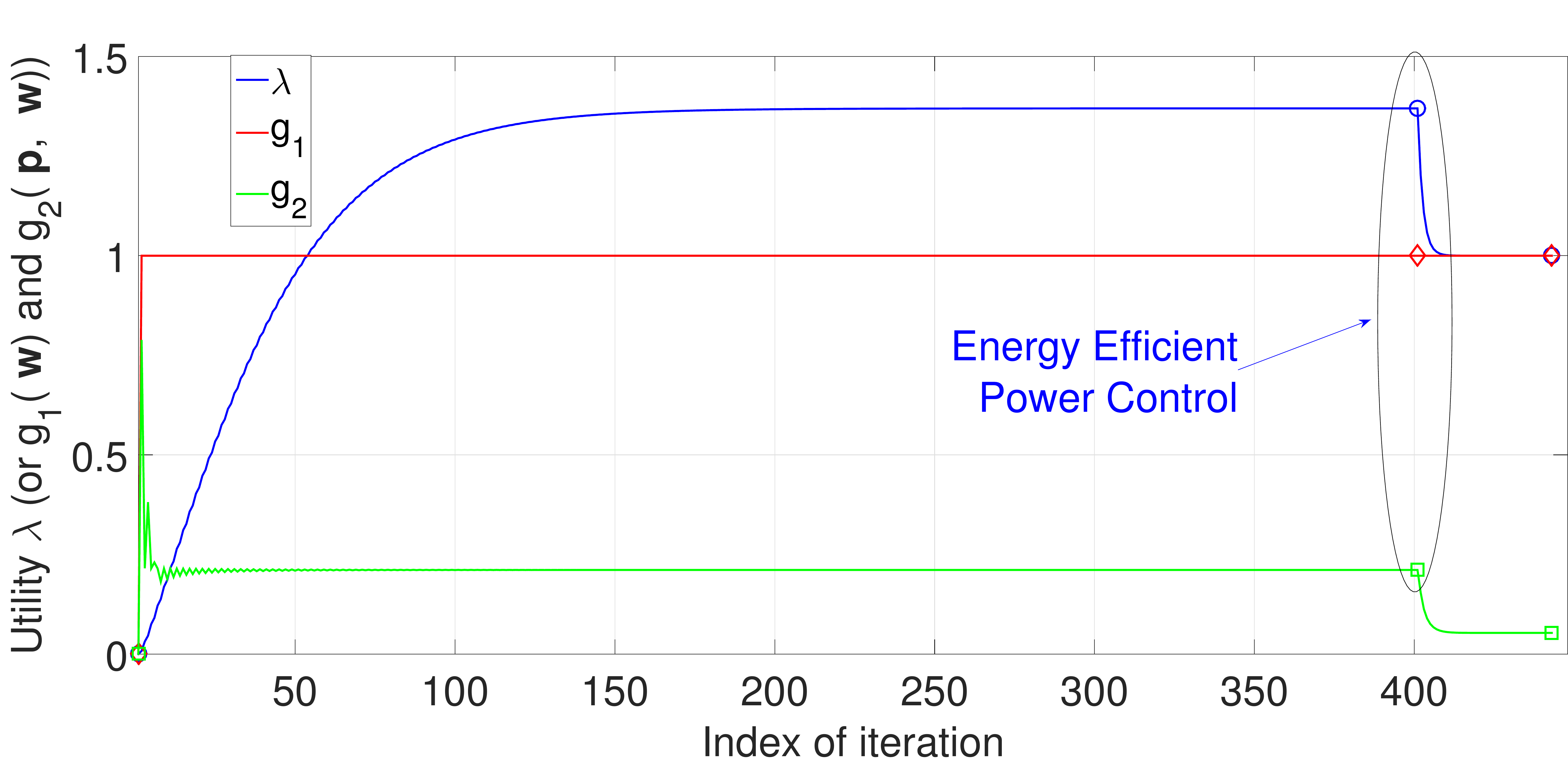}
        \caption{Energy efficient power control.}
				\label{fig:EnergyEff}
    \end{subfigure}
    \caption{Algorithm convergence ($K = 500$, DeUD\_P).}
		\label{fig:AlgorConverge}
\end{figure}

\begin{figure}[t]
    \centering
    \begin{subfigure}[t]{1\columnwidth}
        \centering
        \includegraphics[width=1\columnwidth]{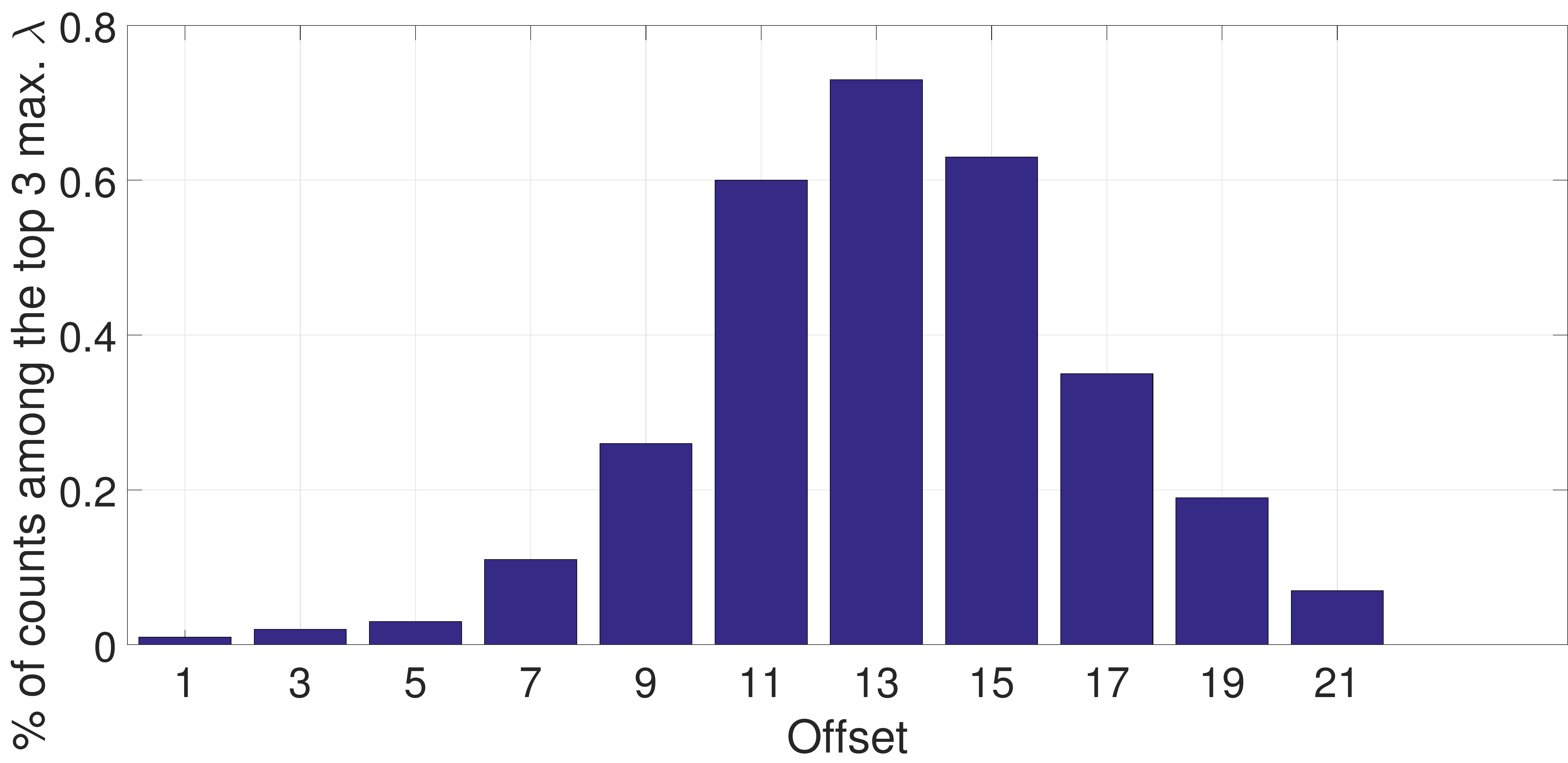}
        \caption{Percentage of counts that the optimized utility with respect to a fixed offset is among the top $3$ maximum values.}
 \label{fig:probselectOff}
    \end{subfigure}
				    \begin{subfigure}[t]{1\columnwidth}
        \centering
        \includegraphics[width=1\columnwidth]{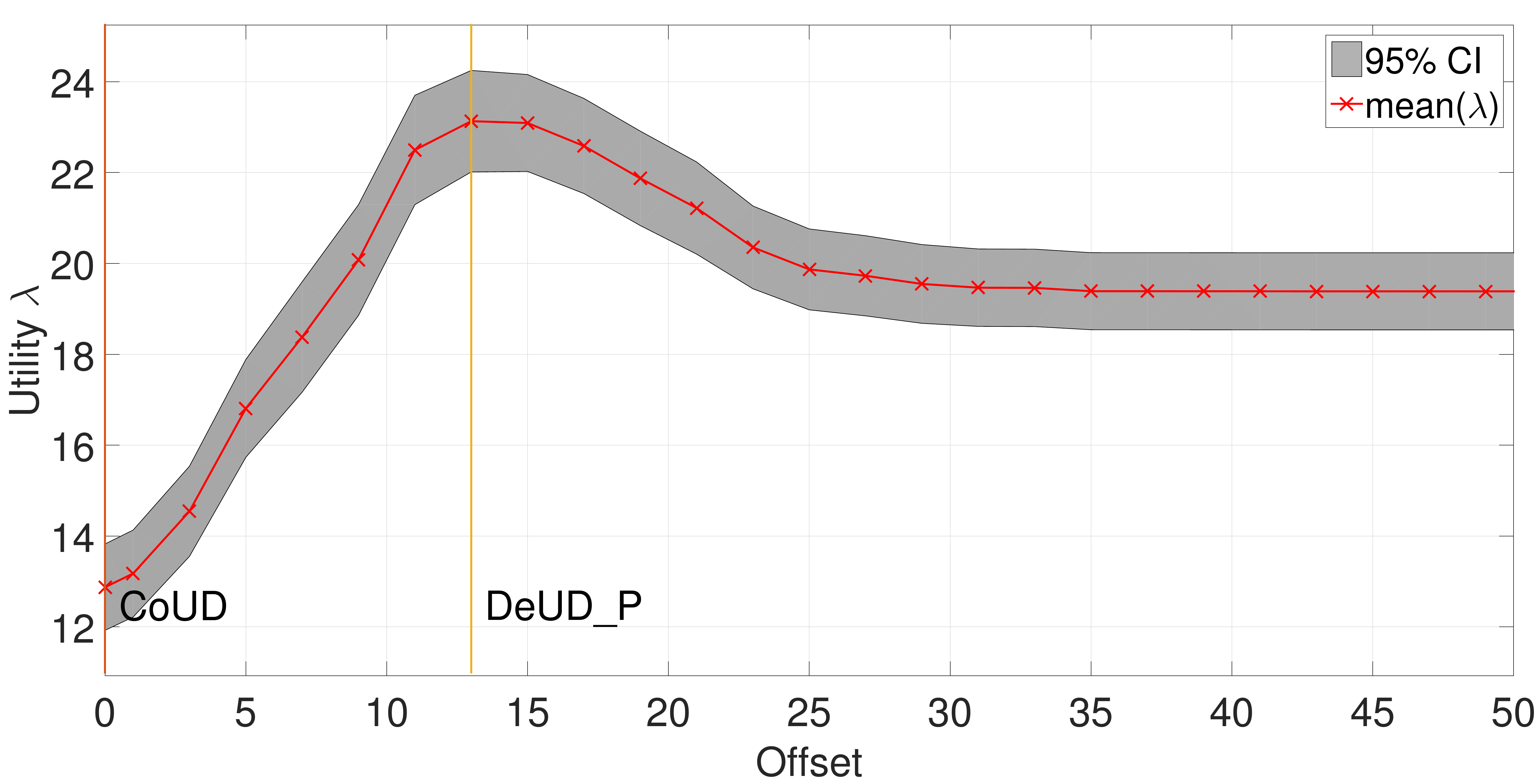}
        \caption{Average utility over 500 tests and the confidence interval for each association policy.}
 \label{fig:UtilityCI_UE100}
    \end{subfigure}
    \begin{subfigure}[t]{1\columnwidth}
        \centering
        \includegraphics[width=1\columnwidth]{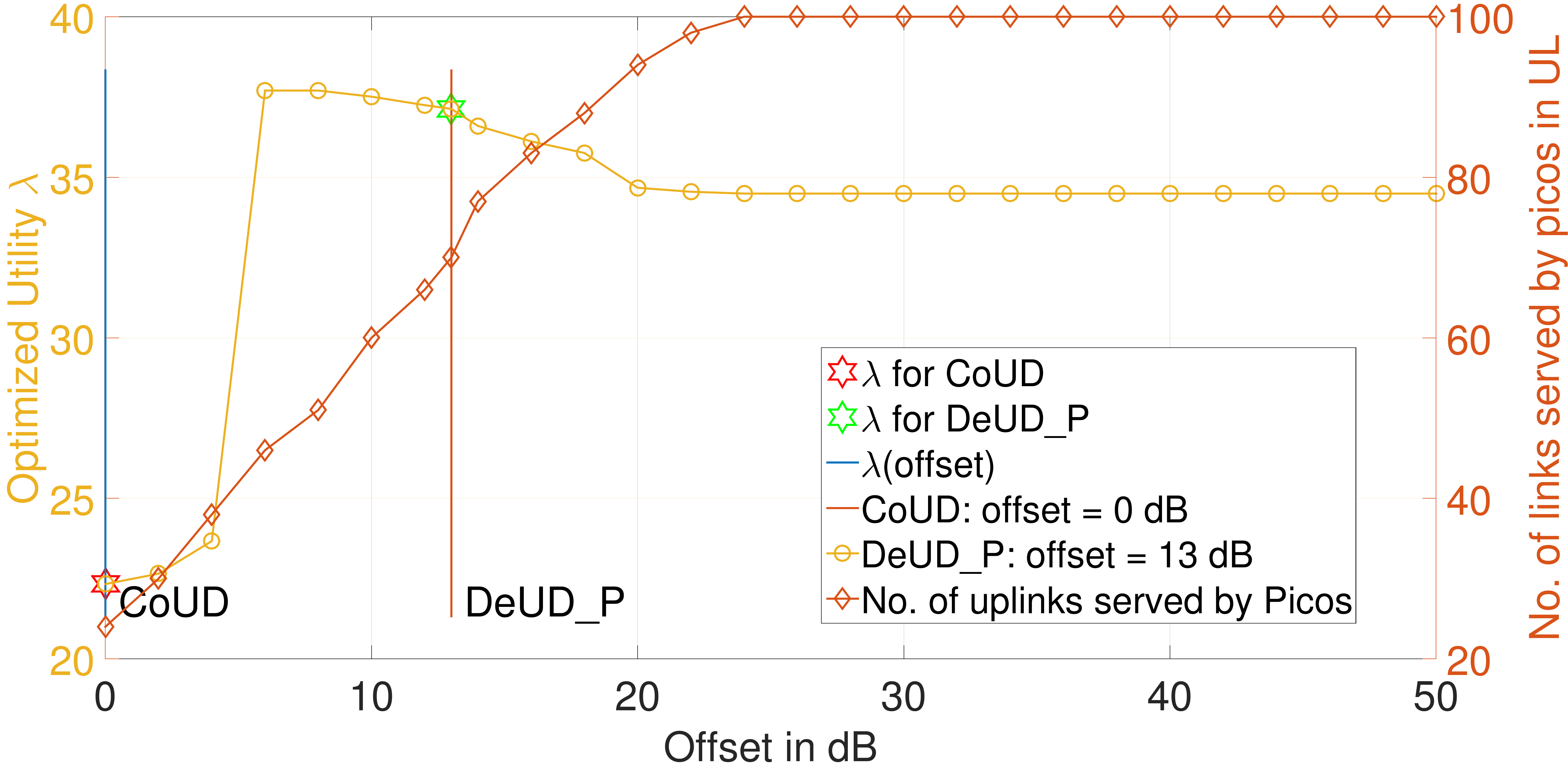}
        \caption{Example trial $\# 1$.}
				 \label{fig:UECurve_1}
    \end{subfigure}
		    \begin{subfigure}[t]{1\columnwidth}
        \centering
        \includegraphics[width=1\columnwidth]{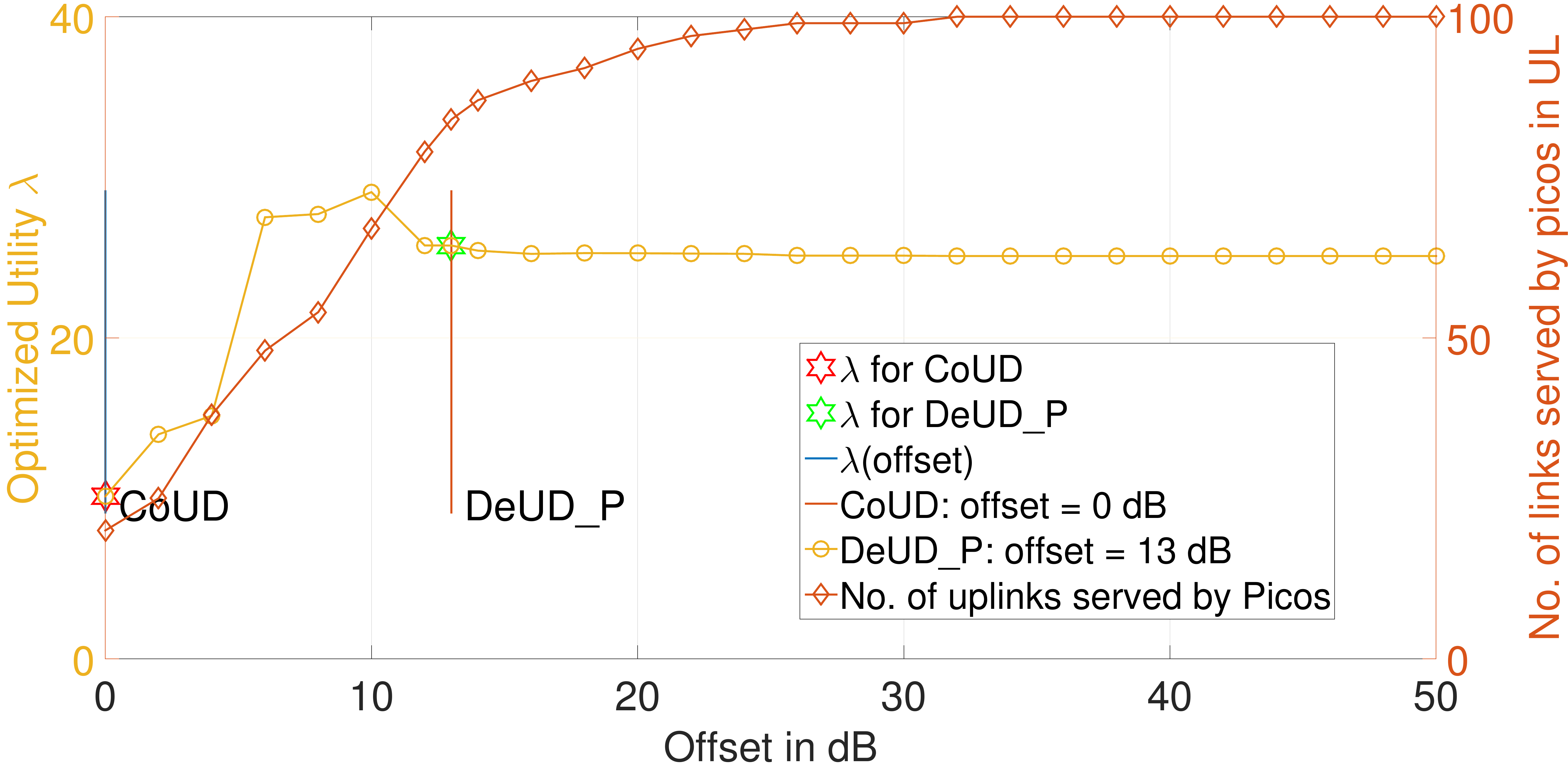}
        \caption{Example trial $\# 2$.}
				\label{fig:UECurve_2}
    \end{subfigure}
    \caption{Optimized utility depending on association policy ($K = 100$).}
		\label{fig:DistributionOffset}
\end{figure}

\begin{figure*}[t]
    \centering
    \begin{subfigure}[t]{2\columnwidth}
        \centering
        \includegraphics[width=1\columnwidth]{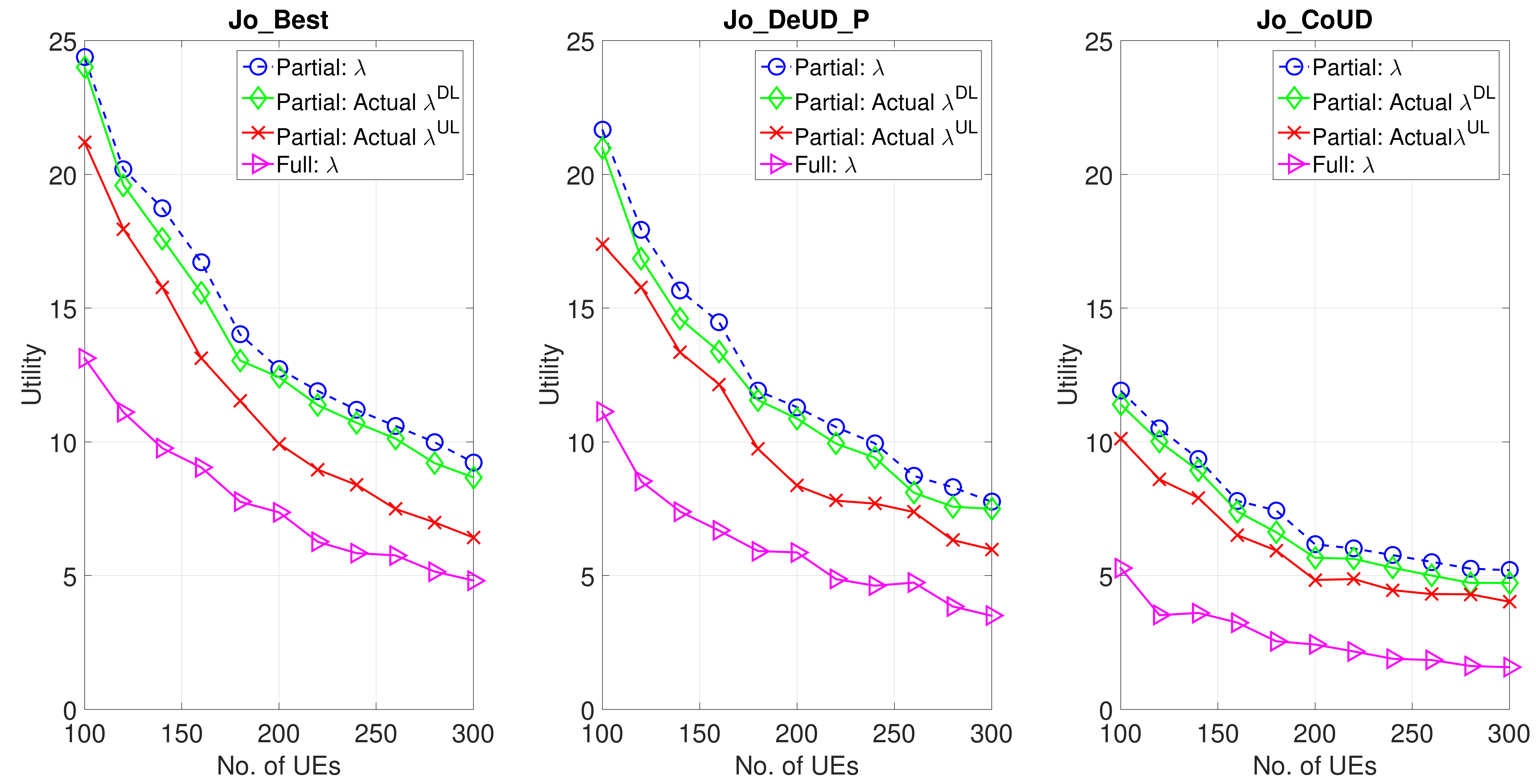}
        \caption{Utility achieved by the joint UL/DL optimization algorithm under different association policies.}
 \label{fig:JoULDL}
    \end{subfigure}
				    \begin{subfigure}[t]{2\columnwidth}
        \centering
        \includegraphics[width=1\columnwidth]{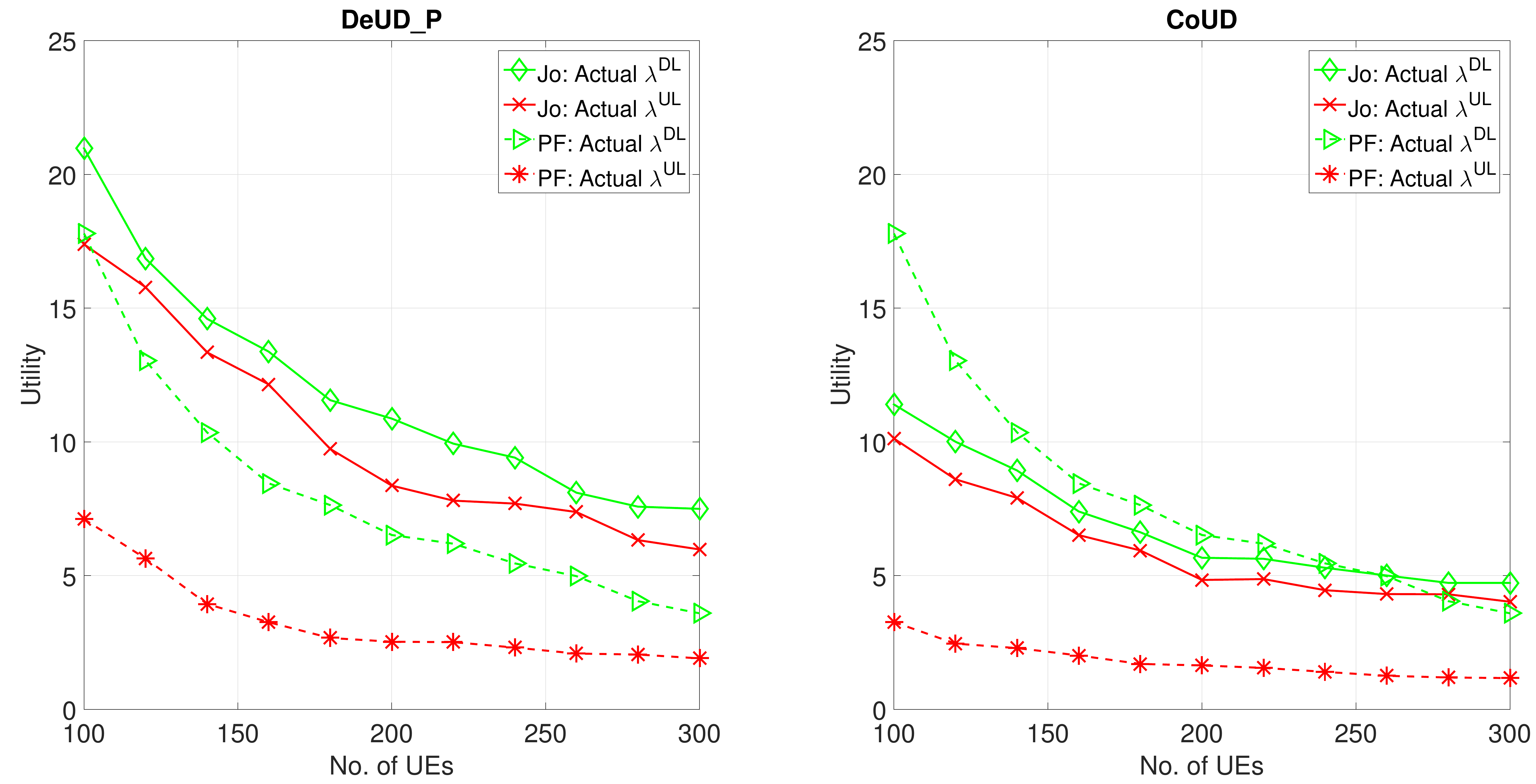}
        \caption{Performance comparison between the joint UL/DL optimization algorithm and the QoS-based PF algorithm under different policies.}
 \label{fig:comparisonPF}
    \end{subfigure}
\caption{Performance evaluation of Algorithm \ref{algo:JointOptimization}.}
\label{fig:JoPerformance}
\end{figure*}

\bibliographystyle{IEEEtran}
\bibliography{References,refs_3}

 \end{document}